\documentclass{lmcs}
\pdfoutput=1
\usepackage[utf8]{inputenc}

\usepackage{lastpage}
\lmcsdoi{21}{1}{29}
\lmcsheading{}{\pageref{LastPage}}{}{}%
{Feb.~23,~2023}{Mar.~27,~2025}{}

\keywords{enumeration, fine-grained complexity, constant delay, union of conjunctive queries, unbalanced triangle detection}

\usepackage{bold-extra}
\usepackage{hyperref}
\usepackage{tikz}
\usetikzlibrary{positioning, quotes}

\def\e#1{\emph{#1}}

\def\scs{\mathbf{S}}
\def\tup#1{\vec{#1}}
\def\cqa{\,\mathbin{\mbox{:-}}\,}
\newcommand{\var}{\operatorname{var}}
\newcommand{\free}{\operatorname{free}}
\newcommand{\atoms}{\operatorname{atoms}}

\newcommand{\calH}{{\mathcal H}}
\newcommand{\graph}[1]{\calH({#1})}

\newcommand{\DelayClin}{{\sf DelayC_{lin}}}
\newcommand{\pEnum}[1]{\text{\sc Enum}\langle#1\rangle}

\newcommand{\BMM}{\textsc{BMM}}
\newcommand{\hyperclique}{\textsc{Hyperclique}}
\newcommand{\fourclique}{4\textsc{-Clique}}
\newcommand{\threesum}{\textsc{3SUM}}
\newcommand{\UTL}{\textsc{VUTL}}
\newcommand{\UTD}{\textsc{VUTD}}
\newcommand{\aUTD}[1]{{#1}\textsc{-VUTD}}

\newcommand{\size}[1]{\lvert{#1}\rvert}
\newcommand{\set}[1]{\lbrace{#1}\rbrace}

\newcommand{\connectors}{\operatorname{connectors}}
\newcommand{\calR}{{\mathcal R}}

\theoremstyle{plain}
\theoremstyle{defC}\newtheorem{exaC}[thm]{Example}

\begin{document}
\title[Unbalanced Triangle Detection and Enumeration Hardness for UCQs]{Unbalanced Triangle Detection and Enumeration Hardness for Unions of Conjunctive Queries}

\author[K.~Bringmann]{Karl Bringmann\lmcsorcid{0000-0003-1356-5177}}[a]
\author[N.~Carmeli]{Nofar Carmeli\lmcsorcid{0000-0003-0673-5510}}[b]
\address{Saarland University and Max Planck Institute for Informatics, Saarland Informatics Campus, Saarbrücken, Germany}
\email{bringmann@cs.uni-saarland.de}
\address{Inria, LIRMM, University of Montpellier, CNRS, Montpellier, France}
\email{nofar.carmeli@inria.fr}

\begin{abstract}
  We study the enumeration of answers to Unions of Conjunctive Queries (UCQs) with optimal time guarantees. More precisely, we wish to identify the queries that can be solved with linear preprocessing time and constant delay. Despite the basic nature of this problem, it was shown only recently that UCQs can be solved within these time bounds if they admit free-connex union extensions, even if all individual CQs in the union are intractable with respect to the same complexity measure. Our goal is to understand whether there exist additional tractable UCQs, not covered by the currently known algorithms.

As a first step, we show that some previously unclassified UCQs are hard using the classic 3SUM hypothesis, via a known reduction from 3SUM to triangle listing in graphs. As a second step, we identify a question about a variant of this graph task that is unavoidable if we want to classify all self-join-free UCQs: is it possible to decide the existence of a triangle in a vertex-unbalanced tripartite graph in linear time? We prove that this task is equivalent in hardness to some family of UCQs. Finally, we show a dichotomy for unions of two self-join-free CQs if we assume the answer to this question is negative.

In conclusion, this paper pinpoints a computational barrier in the form of a single decision problem that is key to advancing our understanding of the enumeration complexity of many UCQs. Without a breakthrough for unbalanced triangle detection, we have no hope of finding an efficient algorithm for additional unions of two self-join-free CQs. On the other hand, a sufficiently efficient unbalanced triangle detection algorithm can be turned into an efficient algorithm for a family of UCQs currently not known to be tractable.
\end{abstract}

\maketitle

\section{Introduction}

Answering queries over relational databases is a fundamental problem in data management.
As the available data in the world grows bigger, so grows the importance of finding the best possible complexity for solving this problem.
Since the query itself is usually significantly smaller than the size of the database, it is common to use \emph{data complexity}~\cite{vardi1982complexity}: we treat the query as fixed, and examine the complexity of finding the answers to the given query over the input database.
As the number of answers to a query may be much larger than the size of the database itself, we cannot hope to generate all answers in linear time in the size of the input. Instead, we use \emph{enumeration} measures.
Since we must read the entire input to verify whether the query has answers, and we must print all answers, the measure of linear preprocessing time and constant delay between two successive answers can be seen as the optimal time complexity for answering queries. The class of queries that can be answered within these time bounds is denoted $\DelayClin$, and recent research asks which queries are in this class~\cite{DurandTutorial20, berkholz2020tutorial}. 

Proving that a query belongs to the class $\DelayClin$ can be achieved by a variety of algorithmic techniques, coupled with insights into the query structure. 
However, proving that a query is unconditionally \emph{not} contained in this class is beyond the state-of-the-art lower bound techniques for the RAM model.
Therefore, one must resort to \emph{conditional lower bounds}. That is, start from a hypothesis on the time complexity of a well-studied problem
and design a reduction from your problem of choice; this proves a lower bound that holds conditional on the starting hypothesis.
While such a conditional lower bound is no absolute impossibility result, it identifies an algorithmic breakthrough that is necessary to find better algorithms for your problem of choice, and thus it yields strong evidence that no better algorithm exists.
This paradigm has been successfully applied in the field of \emph{fine-grained complexity theory} to obtain tight conditional lower bounds for many different problems, see, e.g., \cite{williams2018some,Bringmann19}.
When searching for dichotomies (that characterize which problems in a class admit efficient algorithms), research aiming for lower bounds (conditional or not) has another advantage. The reductions showing hardness often succeed only in some cases. This brings out the other cases, directing us to focus our attention where we have hope for finding efficient algorithms without a major computational breakthrough. This approach has been useful for finding tractable cases that were previously unknown~\cite{DBLP:conf/pods/CarmeliK19}.

When considering \emph{Conjunctive Queries} (CQs) without \emph{self-joins}, the tractability of enumerating query answers with respect to $\DelayClin$ is well understood. The queries with a \emph{free-connex} structure are tractable~\cite{bdg:dichotomy}; these are acyclic queries that remain acyclic with the addition of an atom containing the free variables.
This tractability result is complemented by conditional lower bounds forming a dichotomy:
a \emph{self-join-free} CQ is in $\DelayClin$ if and only if it is free-connex~\cite{bb:thesis,bdg:dichotomy}.
The hardness of cyclic CQs assumes the hardness of finding hypercliques in a hypergraph~\cite{bb:thesis}, while the hardness of acyclic non-free-connex CQs assumes the hardness of Boolean matrix multiplication~\cite{bdg:dichotomy}.
This dichotomy assumes that the CQ does not contain self-joins (that is, every relation appears in at most one atom of the query), which enables assigning different atoms with different relations when reducing a hard problem to query answering.
Not much is known regarding the case with self-joins, but we do know that there are cases where self-joins affect the complexity~\cite{berkholz2020tutorial, carmeli2023selfjoins}.

The next natural class of queries to consider is Unions of Conjunctive Queries (UCQs), which is equivalent to positive relational algebra.
A union of tractable CQs is known to be tractable~\cite{csl/DurandS11}. However, when the union contains an intractable CQ, the picture gets much more complex.
Note that a union that contains an intractable CQ may be equivalent to a union of tractable CQs; in which case, the UCQ is tractable~\cite{DBLP:conf/pods/CarmeliK19}. This can happen for example if the union is comprised of an intractable CQ $Q_1$ and a tractable CQ $Q_2$ subsuming it; then the entire union is equivalent to $Q_2$.
Thus, it makes sense to consider non-redundant UCQs.
It was claimed that a non-redundant UCQ that contains an intractable CQ is necessarily intractable~\cite{icdt/BerkholzKS18}. This claim was disproved in a surprising result showing that a UCQ may be tractable even if it consists solely of intractable CQs~\cite{DBLP:conf/pods/CarmeliK19}.
Specifically, Carmeli and Kröll showed that whenever each CQ in a union can become free-connex (and thus tractable) via a so-called \emph{union extension}, then the UCQ is in $\DelayClin$~\cite{DBLP:conf/pods/CarmeliK19}.
Moreover, every UCQ that we currently know to be in $\DelayClin$ has a free-connex union extension.

In the case of a union of two intractable CQs,
known conditional lower bounds show that
these extensions capture all tractable queries~\cite{DBLP:conf/pods/CarmeliK19}. These lower bounds rely on the same hypotheses as those used for CQs, in addition to a hypothesis on the hardness of detecting a $4$-clique in a graph.
The case of a union of a tractable CQ and an intractable CQ is not yet completely classified, and 
Carmeli and Kröll~\cite{DBLP:conf/pods/CarmeliK19} identified
several open examples, that is,
specific unclassified queries for which the current techniques for an algorithm or a conditional lower bound do not apply.

\paragraph*{Our Contribution}

Our aim is to understand whether there exist additional tractable UCQs, not covered by the currently known algorithms.
We start by showing that some examples of UCQs left open in~\cite{DBLP:conf/pods/CarmeliK19} are hard assuming the standard \emph{3SUM conjecture} (given $n$ integers, it is not possible to decide in subquadratic time whether any three of them sum to~0).
Our reductions go through an intermediate hypothesis that we call Vertex-Unbalanced Triangle Listing (\UTL{};
listing all triangles in an unbalanced tripartite graph requires super-linear time in terms of input and output size).
Specifically, building on a reduction by Fischer, Kaliciak, and Polak~\cite{FischerK024}, we show that the \UTL{} hypothesis is implied by the 3SUM conjecture. We then use \UTL{} to show the hardness of some previously unclassified UCQs.

When trying to reduce \UTL{} to further unclassified UCQs, we identified several issues. 
This led us to introduce a similar hypothesis on Vertex-Unbalanced Triangle Detection (\UTD{}; determining whether an unbalanced tripartite graph contains triangles requires super-linear time in terms of input size).\footnote{Our triangle detection instances are \emph{vertex-unbalanced}, in contrast to a recently formulated hypothesis with the same name that is \emph{edge-unbalanced}~\cite{KopelowitzW20}, see \autoref{appendix:otherUTDhypothesis} for a discussion.}
The \UTD{} hypothesis implies the \UTL{} hypothesis, and thus the former is easier to reduce to UCQs. 
For a discussion of why \UTD{} is a reasonable hypothesis, we refer to \autoref{sec:utd-definition}.
We show that \UTD{} exactly captures the hardness of some family of UCQs that do not have free-connex union extensions: The \UTD{} hypothesis holds if and only if no query in this family is in $\DelayClin$. Thus, determining whether the \UTD{} hypothesis holds is unavoidable if we want to classify all self-join-free UCQs.
Next, we focus on unions of two CQs. We show how, assuming the \UTD{} hypothesis, we can conclude the hardness of any union of one tractable CQ and one intractable CQ that does not have a free-connex union extension. Moreover, if \UTD{} holds, previously known hardness results apply without assuming additional hypotheses. This results in a dichotomy, which is our main result: a union of two self-join-free CQs is in $\DelayClin$ if and only if it has a free-connex union extension, assuming the \UTD{} hypothesis. For these UCQs, we conclude that the currently known algorithms cover all tractable cases that do not require a major breakthrough regarding \UTD{}. 

The main conclusion from our paper is that to make progress in understanding the enumeration complexity of UCQs, it suffices to study the single decision problem of detecting triangles in unbalanced graphs.
More precisely, if we ever find a linear-time algorithm for unbalanced triangle detection, we will also get a breakthrough in UCQ evaluation in the form of an algorithm for some UCQs that do not have a free-connex union extension.
If on the other hand, we assume that there is no linear-time algorithm for unbalanced triangle detection, then for a large class of UCQs (namely, unions of two self-join-free CQs) the currently known algorithms cover all tractable cases.

The paper is organized as follows.
\autoref{sec:Preliminaries} provides basic definitions and results that we use throughout the paper.
In \autoref{sec:utd-definition}, we define \UTL{} and \UTD{}, discuss their connections to other hypotheses, and use them to address some examples of UCQs that were previously unclassified.
In \autoref{sec:equivalence}, we present a family of UCQs that is equivalent in hardness to \UTD{}.
Then, \autoref{sec:ucq-classification} shows the classification of UCQs that we can achieve based on the \UTD{} hypothesis; this proves our main dichotomy result.
The next section contains discussions related to alternative hypotheses.
We discuss how we can conclude the hardness for relaxed time requirements (polylogarithmic instead of constant delay) based on a strengthening of the \UTD{} hypothesis in \autoref{sec:super-const}, and in \autoref {appendix:otherUTDhypothesis} we discuss the difference between \UTD{} and a similar hypothesis for edge-unbalanced graphs that was recently introduced. 
We conclude in \autoref{sec:conclusion}.

\section{Preliminaries} \label{sec:Preliminaries}

\subsubsection*{Databases and Queries}

A \e{relation} is a set of tuples of \e{constants}, where each tuple has the same arity (length).
A \e{schema} $\scs$ is a collection of \e{relation symbols} $R$, each with an associated arity.
A \e{database} $D$ (over the schema $\scs$) associates with each relation symbol $R$ a finite relation, which we denote by $R^D$, whose arity is that of $R$ in $\scs$.

A {\em Conjunctive Query} (CQ) $Q$ over a schema $\scs$ is defined by an expression of the form
$
Q(\tup x) \cqa R_1(\tup{t}_1),\dots,R_n(\tup{t}_n)
$,
where each $R_i$ is a relation symbol of $\scs$, each $\tup{t}_i$ is a
tuple of variables and constants with the same arity as $R_i$, and
$\tup{x}$ is a tuple of variables from
$\tup{t}_1,\dots,\tup{t}_n$.
We usually omit the explicit specification of the schema $\scs$, and assume that it consists of the relation symbols
that occur in the query at hand.
We call $Q(\tup x)$ the \e{head} of $Q$,
and $R_1(\tup{t}_1),\dots,R_n(\tup{t}_n)$ the \e{body} of $Q$. Each
$R_i(\tup t_i)$ is an \e{atom} of $Q$, and the set of all atoms of $Q$ is denoted $\atoms(Q)$.
When the order of the variables in an atom is not important for our discussion, we sometimes denote an atom $R_i(\tup t_i)$ by $R_i(T_i)$ where $T_i$ is a set of variables.
We use $\var(Q)$ to denote the set of variables that occur in $Q$.
We say that $Q$ is \e{self-join-free} if every relation symbol occurs in it at most once. If a CQ is self-join-free, we use $\var(R_i)$ to denote the set of variables that occur in the atom containing $R_i$.
The variables occurring in
the head are called the \e{free variables} and denoted by $\free(Q)$. The variables occurring in the body but not in the head are called \e{existential}.
A \e{homomorphism} $h$ from a CQ $Q$ to a database
$D$ is a mapping of the variables in $Q$ to the constants of $D$, such
that for every atom $R_i(\tup t_i)$ of the CQ, we have that $h(\tup t_i)\in R^D$. 
Each such homomorphism $h$ yields an \emph{answer} $h(\tup x)$ to $Q$.
We denote by $Q(D)$ the set of all answers to $Q$ on $D$.

A \emph{Union of Conjunctive Queries (UCQ)} $Q$ is a finite set of CQs, denoted $Q=\bigcup_{i=1}^{\ell} Q_i$, where $\free(Q_{i})$ is the same for all $1\leq i\leq \ell$.
The set of answers to $Q$ over a database $D$ is the union $Q(D)=\bigcup_{i=1}^{\ell} Q_i(D)$.
Let $Q_1$ and $Q_2$ be CQs.
A \emph{body-homomorphism} from $Q_2$ to $Q_1$ is a mapping $h:\var(Q_2)\rightarrow\var(Q_1)$ such that for every atom $R(\vec{v})$ of $Q_2$, we have that $R(h(\vec{v}))\in Q_1$.
A \emph{body-isomorphism} from $Q_2$ to $Q_1$ is a bijective mapping $h$ such that $h$ is a body-homomorphism from $Q_2$ to $Q_1$ and $h^{-1}$ is a body-homomorphism from $Q_1$ to $Q_2$. We say that two CQs are \emph{body-isomorphic} if there is a body-isomorphism between them.
A \emph{homomorphism} from $Q_2$ to $Q_1$ is a body-homomorphism $h$ such that $h(\free(Q_2))=\free(Q_1)$.
It is well known that $Q_1$ is contained in $Q_2$ (i.e., $Q_1(D)\subseteq Q_2(D)$ on every input $D$) iff there exists a homomorphism from $Q_2$ to $Q_1$~\cite{homomorphism/redundancy}.
We say that a UCQ is \emph{non-redundant} if it does not contain two different CQs such that there is a homomorphism from one to the other.
We often assume that UCQs are non-redundant; otherwise, an equivalent non-redundant UCQ can be obtained by removing CQs.

\subsubsection*{Enumeration Complexity}
An \emph{enumeration problem} $P$ is a collection of pairs $(I,Y)$
where $I$ is an \emph{input} and $Y$ is a finite set of \emph{answers}
for $I$, denoted by $P(I)$. An \emph{enumeration algorithm}
$\mathcal{A}$ for an enumeration problem $P$ is an algorithm that
consists of two phases: \e{preprocessing} and \e{enumeration}. During
preprocessing, $\mathcal{A}$ is given an input \emph{I}, and it may build data structures. During the enumeration phase,
$\mathcal{A}$ can access the
data structures built during
preprocessing, and
it emits the answers $P(I)$, one by one, without repetitions. The time
between printing any two answers during the enumeration phase is
called \emph{delay}.
In this paper, an enumeration problem refers to a query $Q$, the input is a database $D$, and the answer set is
$Q(D)$. Such a problem is denoted $\pEnum{Q}$. 
We adopt \emph{data complexity}, where the query is treated as fixed, and the complexity is with respect to the size of the representation of the database.
We work on the common \emph{Random Access Machine} (RAM) model, where each memory cell stores $\Theta(\log |I|)$ bits, and which supports lookup tables of polynomial size that can be queried in constant time.
The enumeration class $\DelayClin$ is defined as the class of all enumeration problems
that have an enumeration algorithm with $O(|I|)$ preprocessing time and $O(1)$ delay.
Note that this class does not impose a restriction on the memory used.\footnote{Not much is known regarding the complexity of UCQ answering when the memory is restricted. See~\cite[Section 6.3]{DBLP:conf/pods/CarmeliK19} for a related discussion.}
To ease notation, we identify the query with its corresponding enumeration problem, and we denote $Q\in\DelayClin$ to mean $\pEnum{Q}\in\DelayClin$.

\subsubsection*{Hypergraphs}
A {\em hypergraph} $\calH=(V,E)$ is a set $V$ of {\em vertices} and a set $E$ of non-empty subsets of $V$ called {\em hyperedges} (sometimes edges).
Given $S\subseteq V$, the induced subgraph $\calH[S]$ is $(S,E')$ where $E'=\{e\cap S\mid e\in E\}\setminus \{\emptyset\}$.
Two vertices in a hypergraph are {\em neighbors} if they appear in a common edge.
A {\em clique} of a hypergraph is a set of vertices that are pairwise neighbors in $\calH$.
If every edge in $\calH$ has exactly $k$ vertices, we call $\calH$ $k$-{\em uniform}; note that $2$-uniform hypergraphs are just graphs.
For any $\ell > k$, an $\ell$-{\em hyperclique} in a $k$-uniform hypergraph $\calH$ is a set $V'$ of $\ell$ vertices,
such that every subset of $V'$ of size $k$ forms a hyperedge.
A {\em path} of $\calH$ is a sequence of vertices such that every two consecutive vertices are neighbors. The {\em length} of a path $v_1,\ldots,v_n$ is $n-1$.
A {\em simple path} of $\calH$ is a path where every vertex appears at most once.
A {\em chordless path} is a simple path in which no two non-consecutive vertices are neighbors.
A {\em cycle} is a path that starts and ends in the same vertex.
A {\em simple cycle} is a cycle of length $3$ or more where every vertex appears at most once (except for the first and last vertex).
A {\em chordless cycle} is a simple cycle such that no two non-consecutive vertices are neighbors and no edge contains all cycle vertices.
A {\em tetra} of size $k$ (where $k\geq 3$) is a set of $k$ vertices such that every $k-1$ of them are contained in an edge, and no edge contains all $k$ vertices.
A hypergraph is {\em cyclic} if it contains a chordless cycle or a tetra. A hypergraph that is not cyclic is called {\em acyclic} (this is known as $\alpha$-acyclicity)~\cite{brault2016hypergraph}.
Note that ``a 2-uniform hypergraph is cyclic'' is a different way of saying that ``a graph has a cycle''.
A hypergraph is {\em connected} if for any two vertices $u,v$ there is a path starting in $u$ and ending in $v$.
A {\em tripartite graph} is comprised of three sets of vertices $(V_1,V_2,V_3)$ and three sets of edges $E_{1,2}\subseteq V_1\times V_2$, $E_{2,3}\subseteq V_2\times V_3$, and $E_{1,3}\subseteq V_1\times V_3$. A \emph{triangle} in a tripartite graph is a triple of vertices $v_1,v_2,v_3$ such that $(v_1,v_2)\in E_{1,2}$, $(v_2,v_3)\in E_{2,3}$, and $(v_1,v_3)\in E_{1,3}$.

\subsubsection*{Query Structure}
We associate a hypergraph $\calH(Q)=(V,E)$ to a CQ $Q$ where the vertices are the variables of $Q$, and every hyperedge is a set of variables occurring in a single atom of $Q$. That is, $E=\{\{v_1,\ldots,v_\ell\} \mid R_i(v_1,\ldots,v_\ell)\in\atoms(Q)\}$. With slight abuse of notation, we identify atoms of $Q$ with hyperedges of $\calH(Q)$.
A CQ $Q$ is said to be {\em acyclic} if $\calH(Q)$ is acyclic.
Given a CQ $Q$ and a set $S\subseteq\var(Q)$, an {\em $S$-path} is a chordless path $(x,z_1,\ldots,z_k,y)$ in $\calH(Q)$ with $k\geq 1$, such that $x,y\in S$, and $z_1,\ldots,z_k\not\in S$.
A CQ $Q$ is {\em $S$-connex} if it is acyclic and it does not contain an $S$-path~\cite{bdg:dichotomy}.
When referring to a CQ $Q$, we say {\em free-path} for $\free(Q)$-path and {\em free-connex} for $\free(Q)$-connex.
To summarize, every CQ is one of the following:
(1) free-connex;
(2) acyclic and not free-connex, and therefore contains a free-path; or
(3) cyclic, and therefore contains a chordless cycle or a tetra.
We call free-paths, chordless cycles, and tetras {\em difficult structures}. Every CQ that is not free-connex contains a difficult structure.

\subsubsection*{CQ Complexity}
Bagan, Durand, and Grandjean showed that the answers to free-connex CQs can be efficiently enumerated~\cite{bdg:dichotomy}.
This result was complemented by conditional lower bounds showing that other CQs are not in $\DelayClin$, assuming the following hypotheses:

\begin{defi}[\BMM{} Hypothesis]
Two Boolean $n\times n$ matrices cannot be multiplied in time $O(n^2)$.
\end{defi}
\begin{defi}[\hyperclique{} Hypothesis]
For all $k\geq 3$, it is not possible to determine the existence of a $k$-hyperclique in a $(k-1)$-uniform hypergraph with $n$ vertices in time $O(n^{k-1})$.
\end{defi}

Note that the \hyperclique{} hypothesis in particular postulates (for $k=3$) that determining the existence of a triangle in a given graph cannot be solved in time $O(n^2)$.

Boolean matrix multiplication can be encoded in free-paths, and thus self-join-free acyclic CQs that are not free-connex are not in $\DelayClin$, assuming the \BMM{} hypothesis~\cite{bdg:dichotomy}.
The detection of hypercliques can be encoded in tetras and chordless cycles, and thus the first answer to self-join-free cyclic CQs cannot be found in linear time, assuming the \hyperclique{} hypothesis~\cite{bb:thesis}.
As \hyperclique{} implies \BMM{} (\autoref{prop:BMM-UTD}), the known dichotomy can be summarized as:

\begin{thmC}[\cite{bdg:dichotomy,bb:thesis}]\label{theorem:CQdichotomy}
	Let $Q$ be a self-join-free CQ.
	\begin{enumerate}
		\item If $Q$ is free-connex, then $Q\in\DelayClin$.
		\item Otherwise, $Q\not\in\DelayClin$, assuming the \hyperclique{} hypothesis.
	\end{enumerate}
\end{thmC}

We call a CQ {\em difficult} if it matches the last case of \autoref{theorem:CQdichotomy}. Note that a difficult CQ is self-join-free and is either acyclic and not free-connex or cyclic. In other words, a CQ is difficult if it is self-join-free and contains a difficult structure.

\subsubsection*{UCQ Complexity}

The results regarding the tractability of CQs carry over to UCQs if we take into account that CQs in the same union can sometimes ``help'' each other. This is formalized as follows.
Let $Q_s$ be a CQ. We say that $Q_s$ \emph{supplies} a set of variables $V_s\subseteq\var(Q_s)$ if there is $V_s\subseteq S \subseteq \free(Q_s)$ such that $Q_s$ is $S$-connex. Note that when $Q_s$ is free-connex, it supplies any subset of its free variables by taking $S=\free(Q_s)$.

An \emph{extension sequence} for a UCQ $Q^1= Q_1^1\cup\ldots\cup Q_n^1$ is a sequence $Q^1,\ldots,Q^N$ of UCQs where, for all $1< j \le N$, the query $Q^{j}=Q_1^{j}\cup\ldots\cup Q_n^{j}$ is obtained from $Q^{j-1}$ as follows. Given a sequence $\vec{v}_j$ of variables supplied by a CQ $Q_{s(j)}^\ell$ with $\ell\le j$ and $1\le s(j)\le n$, and given a relational symbol $R_j$ that does not appear in $Q^{j-1}$, we have that, for all
$i\in\{1,\dots,n\}$, either 
(1) $Q_i^{j}=Q_i^{j-1}$ or
(2) there is a body-homomorphism $h_{{s(j)},i}$ from $Q_{s(j)}^1$ to $Q_i^1$, and $Q_i^{j}$ is obtained by adding the atom $R_j(h_{{s(j)},i}(\vec{v}_j))$ to $Q_i^{j-1}$.
If such an extension sequence exists, we call $Q^N$ a \emph{union extension} of $Q^1$.

To ease the discussion, we also define the following. We call an atom that appears in a union extension but not in the original query a \emph{virtual atom}. Let $Q_1,Q_2$ be CQs. We say that $Q_2$ \emph{provides} a set of variables $V_1\subseteq\var(Q_1)$ \emph{to} $Q_1$ if
there exist $V_2\subseteq\free(Q_2)$ and a body-homomorphism $h$ from $Q_2$ to $Q_1$ such that $Q_2$ supplies $V_2$ and $h(V_2)=V_1$.
Note that when $Q_2$ provides a set of variables to $Q_1$, we are allowed to extend $Q_1$ in a union extension by an atom containing these variables thanks to $Q_2$.
We say that $Q_2$ provides a difficult structure of $Q_1$ if it provides to $Q_1$ the set of variables that appear in this difficult structure.
We call a union extension \emph{free-connex} if all CQs it contains are free-connex.

\begin{thmC}[\cite{DBLP:conf/pods/CarmeliK19}]\label{theorem:UCQtractability}
	If $Q$ is a UCQ that has a free-connex union extension, then $Q\in\DelayClin$.
\end{thmC}

\begin{exa}
    Let $Q=Q_1\cup Q_2$ with 
\begin{align*}
Q_1(x,y,w)&\leftarrow R_1(x,z),R_2(z,y),R_3(y,w)\text{ and }\\ 	
Q_2(x,y,w)&\leftarrow R_1(x,y),R_2(y,w).
\end{align*}
The CQ $Q_1$ is difficult: it is self-join free, acyclic, and not free-connex, as it contains the free-path $(x,z,y)$.
The CQ $Q_2$ is free-connex, and it supplies the variables $\{x,y,w\}\subseteq\free(Q_2)$. Since there is a body-homomorphism $h:\var(Q_2)\rightarrow\var(Q_1)$ 
with $h((x,y,w))=(x,z,y)$, we can say that $Q_2$ provides the variables $\{x,z,y\}$ or the free-path $(x,z,y)$ to $Q_1$.
Thus, $Q$ has the union extension $Q_1^+\cup Q_2$, where $Q_1^+$ is obtained by adding the virtual atom $R(x,z,y)$ to $Q_1$. Even though $Q_1$ is not free-connex, since the extension is free-connex, we conclude that $Q\in\DelayClin$.
\end{exa}

Existing lower bounds for UCQs rely on the hypotheses used for CQs and on the following:

\begin{defi}[\fourclique{} Hypothesis]
Determining whether a given graph with $n$ vertices contains a $4$-clique has no algorithm running in time $O(n^3)$.
\end{defi}

Note that the \fourclique{} hypothesis is not immediately related to the \hyperclique{} hypothesis, since the \hyperclique{} hypothesis considers (1) on graphs only triangle detection, (2) $4$-clique detection only on 3-uniform hypergraphs, and (3) includes $k > 4$.

\begin{thmC}[\cite{DBLP:conf/pods/CarmeliK19}]\label{thm:two-hards-intractability}
Let $Q$ be a union of two difficult CQs. If $Q$ does not admit a free-connex union extension, then $Q\not\in\DelayClin$, assuming the \hyperclique{} and \fourclique{} hypotheses.
\end{thmC}

Consider a union of two CQs $Q = Q_1 \cup Q_2$. If both $Q_1$ and $Q_2$ are free-connex, then trivially $Q$ has a free-connex union extension, and thus $Q \in \DelayClin$ (by Theorem~\ref{theorem:UCQtractability}). If both $Q_1,Q_2$ are difficult, then $Q \in \DelayClin$ iff $Q$ admits a free-connex union extension (by Theorem~\ref{thm:two-hards-intractability}). 
To classify all unions of two self-join-free CQs, it remains to study queries where $Q_1$ is free-connex, $Q_2$ is difficult, and $Q$ does not have a free-connex union extension.

\section{Unbalanced Triangle Detection and Related Problems}\label{sec:utd-definition}

In this section, we introduce the unbalanced triangle detection hypothesis that is central to this work, and we show its connections to other problems. 
We start our exposition with the well-known 3SUM conjecture and show that it implies a certain hypothesis on triangle listing (VUTL). After demonstrating the usefulness of VUTL by an example, we discuss its shortcomings, which leads us to pose an analogous hypothesis on triangle detection (VUTD). Finally, we demonstrate the usefulness of VUTD by an example and discuss related work.

We start with the classic \threesum{} conjecture from fine-grained complexity theory~\cite{gajentaan1995class}:

\begin{defi}[\threesum{} Conjecture]
Given $n$ integers, deciding whether any three of them sum to 0 has no algorithm running in time $O(n^{2-\varepsilon})$ for any $\varepsilon > 0$.
\end{defi}

We show that the \threesum{} conjecture implies that listing all triangles in an unbalanced tripartite graph requires super-linear time in terms of input and output size.

\begin{description}
 \item[Vertex-Unbalanced Triangle Listing (\UTL{}) Hypothesis] For any constant $\alpha\in$\linebreak[4]$(0,1]$, listing all triangles in a tripartite graph with $|V_3|=n$ and $|V_1|=|V_2|=\Theta(n^\alpha)$
has no algorithm running in time $O(n^{1+\alpha}
+ t)$, where $t$ is the total number of triangles.
\end{description}

\begin{prop}\label{prop:UTL-threesum}
If the \UTL{} hypothesis fails, then the \threesum{} conjecture fails.
\end{prop}

Reductions from 3SUM to triangle listing problems have a long history, as they were pioneered by P\v{a}tra\c{s}cu~\cite{Patrascu10} and further developed in~\cite{DBLP:conf/soda/KopelowitzPP16,ChanH20,WilliamsX20,FischerK024}. We build on this work to prove our reduction. Specifically, we use the following result by Fischer, Kaliciak, and Polak~\cite{FischerK024}. Their result is formulated for the Set Intersection problem, but it can also be viewed as a result for triangle listing (as we will see soon).

\begin{defi}[Set Intersection Problem] 
Given sets $S_1,\ldots,S_N \subseteq \{1,\ldots,U\}$ each of size at most $s$ and a set of $q$
queries $Q \subseteq \{1,\ldots,N\}^2$, compute for each query $(i, j) \in Q$ the set intersection $S_i \cap S_j$.
\end{defi}

\begin{thm}[Set Intersection Hardness, see Theorem 1.6 in~\cite{FischerK024}] \label{thm:FischerKP24}
Let $0 \le \gamma < 1$ and $0 \le \delta \le 1 - \gamma$.
Unless the 3SUM conjecture fails, there is no algorithm for Set Intersection with parameters $|U| = O(n^{1+\delta-\gamma})$, $N = O(n^{(1+\gamma+\delta)/2})$, $s = O(n^{1-\gamma})$, $q = O(n^{1+\gamma})$, and total output size $O(n^{2-\delta})$ that runs in time $O(n^{2-\varepsilon})$, for any $\varepsilon > 0$.
\end{thm}
\begin{proof}[Proof of Proposition~\ref{prop:UTL-threesum}]
An instance of Set Intersection can be viewed as a triangle listing instance by considering the vertex sets $V_1 = V_2 = \{1,\ldots,N\}$ and $V_3 = \{1,\ldots,U\}$, connecting $(i,j) \in V_1 \times V_2$ by an edge if and only if $(i,j) \in Q$, connecting $(i,j) \in V_1 \times V_3$ by an edge if and only if $j \in S_i$, and similarly connecting $(i,j) \in V_2 \times V_3$ by an edge if and only if $j \in S_i$. Computing for each query $(i,j) \in Q$ the set intersection $S_i \cap S_j$ is then the same as listing all triangles containing the edge $(i,j)$, and thus answering all queries is the same as listing all triangles in the graph. 
Theorem~\ref{thm:FischerKP24} thus immediately implies the hardness of triangle listing:
by taking $\gamma = \delta$, we get that, assuming the \threesum{} conjecture, no algorithm lists all triangles in graphs with $|V_1|=|V_2| = O(n^{(1+\gamma+\delta)/2}) = O(n^{\frac{1}{2}+\delta})$ and $|V_3| = O(n^{1+\delta-\gamma}) = O(n)$ containing $O(n^{2-\delta})$ triangles in time $O(n^{2-\varepsilon})$ for any $\varepsilon > 0$. 

Fix a constant $\alpha \in(0,1]$, and assume for the sake of contradiction that it is possible to list all $t$ triangles in a tripartite graph with $|V_1|=|V_2|=\Theta(|V_3|^\alpha)$ in time $O(|V_3|^{1+\alpha}+ t)$. Set $\delta = \min\{\frac{\alpha}{3},\frac{1}{6}\}$. In what follows we show that then we can list all triangles in a graph with $|V_1|=|V_2| = O(n^{\frac{1}{2}+\delta})$ and 
$|V_3| = O(n)$ containing $O(n^{2-\delta})$ triangles in time $O(n^{2-\varepsilon})$ for some $\varepsilon > 0$, which by the last paragraph contradicts the \threesum{} conjecture. This proves that if the \UTL{} hypothesis fails then the \threesum{} conjecture fails.

So suppose we are given a triangle listing instance with $|V_1|=|V_2| = O(n^{\frac{1}{2}+\delta})$ and 
$|V_3| = O(n)$ containing $O(n^{2-\delta})$ triangles. First, we add isolated dummy nodes to ensure $|V_1|=|V_2| = \Theta(n^{\frac{1}{2}+\delta})$ and $|V_3| = \Theta(n)$.
Then we split $V_1$ and $V_2$ each
into $\Theta(n^\delta)$ sets of size $\Theta(n^{\frac{1}{2}})$. This yields $\Theta(n^{2\delta})$ subproblems, each
with $|V_1|=|V_2|=\Theta(n^{\frac{1}{2}})$ and $|V_3|=\Theta(n)$, and their total number of triangles is
$O(n^{2-\delta})$.
Listing the triangles of all subproblems yields the same result as listing the triangles of the original construction.
 
In case $\alpha = \frac{1}{2}$ we are done now. Indeed, each subproblem has
$|V_1|=|V_2|=\Theta(|V_3|^\alpha)$, and since each subproblem can be solved in time
$O(|V_3|^{1+\alpha} + t) = O(n^{1+\alpha} + t)$, the total running time to solve all subproblems
is $O(n^{2\delta+1+\alpha} + n^{2-\delta})$, since there are $\Theta(n^{2\delta})$ subproblems and
their total number of triangles is $O(n^{2-\delta})$.
We can simplify this time bound to $O(n^{2\delta+3/2} + n^{2-\delta}) =
O(n^{2-1/6})$ since $\alpha = \frac{1}{2}$ and $\delta = \frac{\alpha}{3} = \frac{1}{6}$. Since this running time is subquadratic, we obtain the desired contradiction to the \threesum{} conjecture.
 
In case $\alpha < \frac{1}{2}$, we further split $V_1$ and $V_2$ each into $\Theta(n^{\frac{1}{2}-\alpha})$ sets
of size $\Theta(n^\alpha)$. Together with the first splitting step (where we split into
$\Theta(n^{2\delta})$ subproblems), this yields $\Theta(n^{2\delta+1-2\alpha})$ subproblems, each with
$|V_1|=|V_2|=\Theta(n^\alpha)$ and $|V_3|=\Theta(n)$, and their total number of triangles is $O(n^{2-\delta})$.
If each subproblem could be solved in time $O(|V_3|^{1+\alpha} + t) = O(n^{1+\alpha} + t)$, then all
subproblems in total could be solved in time
$O( n^{2\delta+1-2\alpha} \cdot n^{1+\alpha} + n^{2-\delta} )$,
since there are $\Theta(n^{2\delta+1-2\alpha})$ subproblems and their total number of
triangles is $O(n^{2-\delta})$.
We can simplify this time bound to $O(n^{2-\alpha+2\delta} + n^{2-\delta}) =
O(n^{2-\frac{\alpha}{3}})$ since $\delta=\frac{\alpha}{3}$. Since this running time is subquadratic for any fixed constant $\alpha\in(0,\frac{1}{2})$, we again obtain the desired contradiction to the \threesum{} conjecture.

In case $\alpha > \frac{1}{2}$, we split $V_3$ into $\Theta(n^{1-\frac{1}{2\alpha}})$ sets of size
$\Theta(n^{\frac{1}{2\alpha}})$. Together with the first splitting step (where we split into
$\Theta(n^{2\delta})$ subproblems), this yields $\Theta(n^{2\delta+1-\frac{1}{2\alpha}})$ subproblems, each with
$|V_1|=|V_2|=\Theta(n^{\frac{1}{2}})$ and $|V_3|=\Theta(n^{\frac{1}{2\alpha}})$, so $|V_1|=|V_2|=\Theta(|V_3|^\alpha)$. If each
subproblem could be solved in time $O(|V_3|^{1+\alpha} + t)$, then all
subproblems in total could be solved in time
$O( n^{2\delta+1-\frac{1}{2\alpha}} \cdot |V_3|^{1+\alpha} + n^{2-\delta} )$,
since there are $\Theta(n^{2\delta+1-\frac{1}{2\alpha}})$ subproblems and their total number of
triangles is $O(n^{2-\delta})$. Plugging in $|V_3|=\Theta(n^{\frac{1}{2\alpha}})$ yields time
$O( n^{2\delta+1-\frac{1}{2\alpha} + \frac{1+\alpha}{2\alpha}} + n^{2-\delta} ) = O(n^{3/2+2\delta} + n^{2-\delta}) =
O(n^{2-1/6})$ since $\delta=\frac{1}{6}$, which again yields the desired contradiction.
\end{proof}

Using \UTL{}, some UCQs that were left open by prior work are not in $\DelayClin$.

\begin{exaC}[{\cite[Example 44]{DBLP:conf/pods/CarmeliK19}}]\label{example:cyclic-guarded-hard}
Let $Q=Q_1\cup Q_2$ with
\begin{align*}
	Q_1(x,z,y,v)&\cqa R_1(x,z,v),R_2(z,y,v),R_3(y,x,v)\text{ and}\\
	Q_2(x,z,y,v)&\cqa R_1(x,z,v),R_2(y,t_1,v),R_3(t_2,x,v).
\end{align*}
Note that $Q_2$ is free-connex (and so tractable on its own), while $Q_1$ is cyclic (and so intractable on its own).
The only difficult structure in $Q_1$ is the cycle $x,y,z$.
If the cycle were provided by $Q_2$, we would be able to eliminate the cycle via an extension by adding to $Q_1$ a virtual atom with the cycle variables. Such an extension would be free-connex, entailing the tractability of $Q$. However, $y$ is not provided (as no free variable of $Q_2$ maps to $y$ in the homomorphism from $Q_2$ to $Q_1$), and so the currently known algorithm cannot be applied.
The existing approach to show the difficulty of a CQ with a cycle is to encode the triangle finding problem to this cycle.
We assign the variables $x$, $y$, and $z$ with the vertices of the graph, while $v$ is always assigned a constant $\bot$.
That is, for every edge $(u,v)$ in the input graph, we include the tuple $(u,v,\bot)$ in all three relations. Then, $Q_1$ returns all tuples $(a,b,c,\bot)$ such that $(a,b,c)$ is a triangle.
However, in our case, such an encoding may result in $n^3$ answers to $Q_2$ given a graph with $n$ vertices. This means that if the input graph has triangles, we are not guaranteed to find one in $O(n^2)$ time by evaluating the union efficiently, and we do not obtain a contradiction to the \hyperclique{} hypothesis.
By using \emph{unbalanced} tripartite graphs (where one vertex set is larger than the other two), we can make use of the fact that $y$ is not provided to show hardness.
We encode triangle finding to our databases similarly to before, except we make sure to assign the large vertex set to $y$, while $x$ and $z$ are assigned vertex sets of size $n^\alpha$. This way, while $Q_1$ finds the triangles in the graph, $Q_2$ has at most $n^{3\alpha}$ answers (this happens because all variables mapping to the large vertex set in $Q_2$ are existential). Assuming $Q\in\DelayClin$, we can compute all answers over such a construction in $O(n^{1+\alpha}+n^{3\alpha}+t)$ time.
If we take $\alpha\leq \frac{1}{2}$, this is time $O(n^{1+\alpha}+t)$, contradicting the \UTL{} hypothesis and thus also the \threesum{} conjecture.
\qed
\end{exaC}

In \autoref{example:cyclic-guarded-hard}, we can use a triangle listing hypothesis because the variables of the cycle in $Q_1$ are free. However, there exist similar examples where some of these variables are existential. In these cases, we can use a similar argument if we start from triangle detection instead of triangle listing. This leads us to introduce the following hypothesis.

\begin{description}
 \item[Vertex-Unbalanced Triangle Detection (\UTD{}) Hypothesis]
For any $\alpha\in(0,1]$, determining whether there exists a triangle in a tripartite graph with $|V_3|=n$ and $|V_1|=|V_2|=\Theta(n^\alpha)$
has no algorithm running in time $O(n^{1+\alpha})$.
\end{description}

\begin{figure}
\centering
\begin{tikzpicture}[every edge quotes/.style = {auto, font=\footnotesize, color=gray, sloped}]
\node (VUTD) {\UTD{}};
\node (VUTL) [above right=of VUTD] {\UTL{}};
\node (3SUM) [above left=of VUTD] {\threesum{}};
\node (sVUTD) [below=of 3SUM] {sVUTD};
\node (Hyperclique) [below right=of VUTL] {\hyperclique{}};
\node (dummy) [above right=of Hyperclique] {};
\node (BMM) [below right=of dummy] {\BMM{}};
\draw (3SUM.east) edge[->, "Proposition~\ref{prop:UTL-threesum}"] (VUTL.west);
\draw (VUTD.east) edge[->, "immediate"] (VUTL.south);
\draw (VUTD.east) edge[->, "Proposition~\ref{prop:hyperclique-UTD}"{xshift=2pt}] (Hyperclique.west);
\draw (Hyperclique.east) edge[->, "Proposition~\ref{prop:BMM-UTD}"] (BMM.west);
\draw (sVUTD.east) edge[->, "immediate"{yshift=2pt}] (VUTD.west);
\end{tikzpicture}
\caption{Connections between the hypotheses mentioned in this paper. An edge $H_1\rightarrow H_2$ means that $H_1$ implies $H_2$. The hypotheses EUTD and \fourclique{} are not in the figure since we do not know of connections between them and the other hypotheses.} \label{fig:hypotheses}
\end{figure}
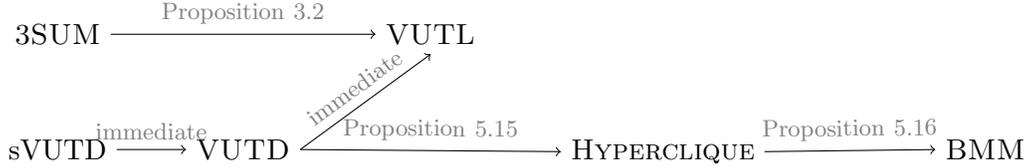

The \UTD{} hypothesis implies the \UTL{} hypothesis. See Figure~\ref{fig:hypotheses} for a summary of the connections between the hypotheses mentioned in this paper.
Unlike \UTL{}, the \UTD{} hypothesis cannot only be used when a CQ in the union contains a cycle, but also when it contains a free-path.
The following example, also left open by prior work, illustrates this case.

\begin{exaC}[{\cite[Example 35]{DBLP:conf/pods/CarmeliK19}}]\label{example:separated}
Let $Q=Q_1\cup Q_2$ with
\begin{align*}
	Q_1(x,y,w)&\cqa R_1(x,z),R_2(z,y),R_3(y,w)\text{ and}\\
	Q_2(x,y,w)&\cqa R_1(x,t_1),R_2(t_2,y),R_3(w,t_3).
\end{align*}
The only difficult structure in $Q_1$ is the free-path $x,z,y$, while $Q_2$ is free-connex.
If the free-path were provided by $Q_2$, we would be able to eliminate the free-path by adding to $Q_1$ a virtual atom with the free-path variables, resulting in a free-connex extension, and entailing the tractability of $Q$. However, $z$ is not provided.
The existing approach to show the difficulty of a CQ with a free-path is to encode the Boolean matrix multiplication problem to this path. However, in our case, such an encoding may result in $n^3$ answers to $Q_2$, so evaluating the union efficiently is not guaranteed to find all non-zero entries in the multiplication result in $O(n^2)$ time, and this would not contradict \BMM{}.
By using \emph{unbalanced} tripartite graphs, we can use the fact that $z$ is not provided to show hardness.
We assign the large vertex set to $z$, while $x$ and $y$ are assigned vertex sets of size $n^\alpha$, and $w$ is assigned a constant $\bot$.
Under this construction, $Q_1$ returns tuples $(a,b,\bot)$ such that some vertex $c$ is a neighbor to both $a$ and $b$. For every such answer, we check whether $a$ and $b$ are neighbors. If they are, we determine that a triangle exists. 
In this way, $Q_1$ finds all candidates for triangles in the graph.
Since no free variable can be assigned values from the large vertex set, $Q_1$ and $Q_2$ have at most $n^{3\alpha}$ answers each. Assuming $Q\in\DelayClin$, we can compute all answers in time $O(n^{1+\alpha}+n^{3\alpha})$. If we take $\alpha\leq\frac{1}{2}$, this time is $O(n^{1+\alpha})$,
contradicting the \UTD{} hypothesis.
\qed
\end{exaC}

\autoref{example:separated} demonstrates that if we assume the \UTD{} hypothesis, we can prove the hardness of previously unclassified UCQs. 
However, unlike the similar listing problem \UTL{}, we are not aware of a complexity conjecture as established as \threesum{} that implies the hardness of \UTD{}. Let us comment on why \threesum{} can be reduced to \UTL{} but not \UTD{}: The reduction from \threesum{} to \UTL{} of Proposition~\ref{prop:UTL-threesum} introduces many false positives, that is, each \threesum{} solution generates a triangle, but also some non-solutions generate a triangle. By listing all triangles we can filter out false positives to then solve \threesum{}. This reduction does not work for \UTD{}, because by only detecting a triangle we cannot remove the false positives.

In the following, we argue that the \UTD{} hypothesis to the very least formalizes a computational barrier that is hard to overcome, and we discuss reasons to suspect the hypothesis holds.
The state of the art for triangle detection relies on matrix multiplication: Compute the matrix product of the adjacency matrix of $V_1 \times V_3$ with the adjacency matrix of $V_3 \times V_2$ to obtain all pairs $(v_1,v_2)$ connected by a 2-path, and then check for each such pair whether it also forms an edge in the graph.
This classic algorithm by Itai and Rodeh~\cite{ItaiR78} has not been improved since 1978, which is not for lack of trying.
For $\alpha=1$ this algorithm runs in time $O(n^\omega)$, where $\omega < 2.373$ is the exponent of matrix multiplication. 
While some researchers believe that $\omega$ should be~2, it was shown that the current matrix multiplication techniques cannot reach this time bound~\cite{AlonSU13,AlmanVW18,AmbainisFG15,VWilliams19}. Thus, if $\omega$ is~$2$, a significant breakthrough is needed to prove that.
Moreover, since $\omega$ is defined as an infimum, even $\omega=2$ does not mean that matrix multiplication is in time $O(n^2)$, for instance, an $O(n^2 \log n)$-time algorithm would also show that $\omega$ is $2$. 
Finally, over the last 35 years $\omega$ has seen only a small improvement from $2.3755$~\cite{CoppersmithW90} to $2.3716$~\cite{WilliamsXXZ24}.
In summary, quadratic-time matrix multiplication seems very far away, if not impossible. Since the best-known algorithm for triangle detection uses matrix multiplication, we see this as a reason to suspect that the \UTD{} hypothesis holds.
Here we focused on the case $\alpha = 1$, but the same discussion also applies to $\alpha < 1$; in this case, the fastest known running time for the corresponding matrix multiplication is of the form $O(n^{1+\alpha+\varepsilon_\alpha})$, where $\varepsilon_\alpha > 0$ is a constant depending only on $\alpha$~\cite{GallU18}.

In this section, we phrased the \UTD{} hypothesis, discussed its connection to related problems, and showed that it can be used in some cases to show the hardness of UCQs. In the next section, we show that determining that the \UTD{} hypothesis does not hold would also have implications for UCQs, as it would identify currently unclassified tractable UCQs. In particular, \autoref{sec:equivalence} proves that some family of UCQs is equivalent to \UTD{}.

\section{Equivalence of \UTD{} and a Family of UCQs}\label{sec:equivalence}

In this section, we show an equivalence between unbalanced triangle detection and the evaluation of a family of UCQs. As a result, we obtain that if the \UTD{} hypothesis does not hold, then free-connex union extensions do not capture all UCQs in $\DelayClin$.
We prove the following theorem.
\begin{thm}\label{thm:exact-example}
  There exists a family of UCQs with no free-connex union extensions such that the \UTD{} hypothesis holds if and only if no query of the family is in $\DelayClin$.
\end{thm}
To prove \autoref{thm:exact-example}, we need to be more specific about the values of $\alpha$ for which we assume that \UTD{} holds. For this reason, we define the following hypothesis for a fixed $\alpha$.

\begin{description}
\item[\aUTD{\boldmath$\alpha$} Hypothesis] Determining whether there exists a triangle in a tripartite graph with $|V_1|=|V_2|=n^\alpha$ and $|V_3|=n$ has no algorithm running in time $O(n^{1+\alpha})$.
\end{description}

Then, the \UTD{} hypothesis is that \aUTD{$\alpha$} holds for every constant $\alpha\in(0,1]$.
We next show that \aUTD{$\alpha$} is ``monotone'' in the sense that it implies \aUTD{$\beta$} for larger values of $\beta$.

\begin{prop}\label{prop:UTD-monotone}
If \aUTD{$\alpha$} holds then \aUTD{$\beta$} holds for all $\beta \ge \alpha$.
\end{prop}
\begin{proof}
We show a self-reduction that splits the set $V_3$. 
Let $0<\alpha<\beta\le 1$, and assume that determining whether there exists a triangle in a tripartite graph with $|V_1| = |V_2| = \Theta(n^\beta)$ and $|V_3| = n$  has an $O(n^{1+\beta})$-time algorithm. That is, we assume that \aUTD{$\beta$} fails and want to prove that \aUTD{$\alpha$} fails. 
To this end, let $G = (V_1 \cup V_2 \cup V_3,E)$ be a tripartite graph with $|V_1| = |V_2| = \Theta(n^\alpha)$ and $|V_3| = n$. Split $V_3$ into $\Theta(n^{1-\alpha/\beta})$ subsets of size
$\Theta(n^{\alpha/\beta})$. This splits $G$ into $\Theta(n^{1-\alpha/\beta})$ subgraphs
$G_1,\ldots,G_t$. Each subgraph is tripartite with parts $V_1,V_2,V_3'$ with $|V_1|=|V_2|=\Theta(n^\alpha)=\Theta({|V_3'|}^\beta$). Therefore, the
assumed algorithm determines whether $G_i$ has a triangle in time
$O({|V_3'|}^{1+\beta})$. Running this algorithm on each graph $G_i$ takes total time
$O(n^{1-\alpha/\beta} {|V_3'|}^{1+\beta}) = O(n^{1-\alpha/\beta + \alpha/\beta + \alpha})
= O(n^{1+\alpha})$. Thus, we can solve the given \aUTD{$\alpha$} instance $G$ in time
$O(n^{1+\alpha})$.
\end{proof}

We prove \autoref{thm:exact-example} with the following family of UCQs.

\begin{exa}\label{example:exact}
For any integer $c \ge 1$, consider the union $Q_{[c]}$ containing the following CQs.
\begin{align*}
Q_1(v_1,\ldots,v_{2c})\cqa& R_1(x,y),R_2(y,z),R_3(x,z),R_4(v_1,\ldots,v_{2c}),\\
&\;S_1(x),\ldots,S_c(x),T_1(y),\ldots,T_c(y)\\
Q_2(v_1,\ldots,v_{2c})\cqa& S_1(v_1),\ldots,S_c(v_c),T_1(v_{c+1}),\ldots,T_c(v_{2c})\\
Q_3(v_1,\ldots,v_{2c})\cqa& R_1(v_1,t_1),R_2(t_2,v_2),R_4(t_3,t_4,v_3,\ldots,v_{2c})\\
Q_4(v_1,\ldots,v_{2c})\cqa& R_1(t_1,v_1),R_2(t_2,v_2),R_4(t_3,t_4,v_3,\ldots,v_{2c})
\end{align*}

Note that $Q_{[c]}$ does not have a free-connex union extension.
Indeed, $Q_1$ contains a cycle $x,y,z$. Since no other CQ in the union provides all three cycle variables (in the body-homomorphisms from $Q_2$, $Q_3$ and $Q_4$ to $Q_1$, no free variables are mapped to $z$, $y$ and~$x$, respectively), any union extension of $Q_1$ preserves this cycle.

\begin{clm}
 If \UTD{} does not hold, then $Q_{[c]}\in\DelayClin$ for all sufficiently large $c$.
\end{clm}
\begin{proof}
 If \UTD{} does not hold, then \aUTD{$\beta$} does not hold for some $\beta\in(0,1)$.  According to \autoref{prop:UTD-monotone}, \aUTD{$\alpha$} does not hold for all $\alpha<\beta$.
 That is,
for all $\alpha\in(0,\beta)$, determining whether there exists a triangle in a tripartite graph with $|V_1|=|V_2|=n^\alpha$ and $|V_3|=n$ can be done in time $O(n^{1+\alpha})$.
 Let $c\geq \frac{1}{\beta}+1$.
 We show how, given a database instance $D$, we can enumerate $Q_{[c]}(D)$ with linear preprocessing and constant delay. 

First note that in each of $Q_2$, $Q_3$, and $Q_4$, every variable only appears in one atom, and so they are free-connex.
Thus, we can compute $Q_2(D)$, $Q_3(D)$ and $Q_4(D)$ with linear preprocessing and constant delay each. In the following, we show how to find $Q_1(D)$ with constant delay after $O(|D|+|Q_2(D)|+|Q_3(D)|+|Q_4(D)|)$ preprocessing time. This means that by interleaving the computation of the preprocessing of $Q_1$ with the evaluation of the other CQs, we can enumerate the answers to $Q_1$ with constant delay directly after the end of the enumeration of the other CQs.
More formally, according to the ``Cheater's Lemma''~\cite[Lemma 7]{DBLP:conf/pods/CarmeliK19}, since the delay between answers is constant except for at most four times where it is linear (corresponding to the preprocessing required for each CQ), and since there are at most four duplicates per answer (as an answer can be obtained from each of the CQs), the algorithm we present here can be adjusted to work with linear preprocessing time and constant delay with no duplicates.

Note that if one of the relations of $Q$ is empty, then $Q_1(D)=\emptyset$, and we can finish the evaluation of $Q_1(D)$ immediately. In the following, we assume that no relation is empty.
Consider the Boolean query $Q_1'()$ with the same body as $Q_1$.
Note that $Q_1(D)$ is exactly $R_4^D$ if $Q_1'$ evaluates to true, and it is empty otherwise.
To evaluate $Q_1'$, we can first filter (in linear time) the relations $R_1^D$, $R_2^D$, and $R_3^D$ by performing semi-joins with $S_i^D$ and $T_i^D$ for all $i$.
Formally, we set
 \begin{align*}
 E_{1,2}&=\set{(a,b)\mid R_1^D(a,b)\wedge \forall{i\in [c]}: S_i^D(a)\wedge T_i^D(b)}\text{,}
 \\
 E_{2,3}&=\set{(b,d)\mid R_2^D(b,d)\wedge \forall{i\in [c]}: T_i^D(b)}\text{, and}
 \\
 E_{1,3}&=\set{(a,d)\mid R_3^D(a,d)\wedge \forall{i\in [c]}: S_i^D(a)}\text{.} 
 \end{align*}
Now it is enough to evaluate $Q_1''()\cqa E_{1,2}(x,y),E_{2,3}(y,z),E_{1,3}(x,z)$ since $Q_1'()=Q_1''()$.
Denote
 \begin{align*}
 V_1&=\set{a\mid \exists{b}: E_{1,2}(a,b)}\text{,}\\
 V_2&=\set{b\mid \exists{a}: E_{1,2}(a,b)}\text{, and}\\
 V_3&=\set{d\mid \exists{b}: E_{2,3}(b,d)}\text{.}
 \end{align*}
Note that every $a \in V_1$ satisfies $S_i^D(a)$ for all $i \in [c]$. Thus, we have $|V_1| \le |S_i^D|$,  and similarly $|V_2| \le |T_i^D|$, for all $i \in [c]$. Since $|Q_2(D)| = \prod_{i \in [c]} |S_i^D| \cdot |T_i^D|$ we obtain $|Q_2(D)| \ge (|V_1||V_2|)^c$. Moreover, since we assume all relations to be non-empty, we have $|Q_3(D)| \ge |\{ a \mid \exists b\colon R_1^D(a,b) \}| \cdot |\{ u \mid \exists v\colon R_2^D(v,u) \}| \ge |V_1| \cdot |V_3|$, and similarly $|Q_4(D)| \ge |V_2| \cdot |V_3|$.

If $|V_3|\leq \max\{|V_1|,|V_2|\}^{c-1}$, then we evaluate $Q_1''$ in $O(|V_1||V_2||V_3|)$ time by checking all possible assignments to $x$, $y$ and $z$. Since $|V_1||V_2||V_3|\leq (|V_1||V_2|)^c\leq |Q_2(D)|$, this takes time $O(|Q_2(D)|)$.
The second case is $|V_3|>\max\{|V_1|,|V_2|\}^{c-1}$.
We set $n=|V_3|$ and $\alpha=\log_n{\max\{|V_1|,|V_2|\}}$; note that $\alpha<\frac{1}{c-1}\leq\beta$.
We fill up the smaller of $V_1,V_2$ with dummy vertices to ensure $|V_1| = |V_2| = \Theta(|V_3|^\alpha)$.
Applying an $O(n^{1+\alpha})$-time triangle detection algorithm to the graph given by the edges $E_{1,2}$, $E_{2,3}$ and $E_{1,3}$ answers $Q_1''$ in time $O(|V_3|^{1+\alpha}) = O(|V_3| \cdot |V_1| + |V_3|\cdot |V_2|)$.
Note that $|Q_3(D)|\geq \size{V_1}\size{V_3}$ and $|Q_4(D)|\geq\size{V_2}\size{V_3}$, so this running time is $O(|Q_3(D)| + |Q_4(D)|)$. 
If $Q_1''$ evaluates to false, $Q_1$ returns no answers, and we are done; otherwise, we output $R_4^D$ with constant delay.
In total, this finds $Q_1(D)$ with constant delay after $O(|D|+|Q_2(D)|+|Q_3(D)|+|Q_4(D)|)$ preprocessing time.
\end{proof}

Note that as part of the proof of this claim, we showed that $Q_{[c]}$ is in $\DelayClin$ in case $|V_3|\leq \max\{|V_1|,|V_2|\}^{c-1}$ without relying on any assumption. This demonstrates that $z$ must have a large domain for this query not to be in $\DelayClin$.
That is, $Q_{[c]}$ is not in $\DelayClin$ (assuming the \UTD{} hypothesis) only when we can make no additional assumptions on the instance; if the domain of $z$ is limited, the query may become easy.
This also shows that in any construction that proves a lower bound for $Q_{[c]}$, we must assign $z$ with a larger domain than that of the other variables. Indeed, we do this in the proof of the following claim.

 \begin{clm}
If \UTD{} holds, then $Q_{[c]}\not\in\DelayClin$ for all $c$.
 \end{clm}

\begin{proof}
Assume by contradiction that $Q_{[c]}\in\DelayClin$ for some $c$.
We start with a tripartite graph $G$ with $V_1$,$V_2$,$V_3$, $E_{1,2}$, $E_{2,3}$ and $E_{1,3}$, where $|V_1|=|V_2|=n^\alpha$ and $|V_3|=n$ for some $n\in\mathbb{N}$ and $\alpha\leq\frac{1}{2c-1}$.
We construct a database instance $D$ as follows:
We assign $R_1^D=E_{1,2}$, $R_2^D=E_{2,3}$, $R_3^D=E_{1,3}$, and $R_4^D=\set{(\bot,\ldots,\bot)}$. For all $i\in[c]$, we assign $S_i^D=V_1$ and $T_i^D=V_2$.
The answers $Q_1(D)$ consist of $(\bot,\ldots,\bot)$ if there is a triangle in $G$ and no answers otherwise.
As for the other CQs, 
$\size{Q_2(D)}=(|V_1||V_2|)^c$, $\size{Q_3(D)}=|V_1||V_3|$, and $\size{Q_4(D)}=|V_2||V_3|$.
The tuple $(\bot,\ldots,\bot)$ is not an answer to CQs other than $Q_1$, so $(\bot,\ldots,\bot)\in Q_{[c]}(D)$ if and only if there is a triangle in $G$.
If $Q_{[c]}\in\DelayClin$, then we can compute all of $Q_{[c]}(D)$ in time $O((|V_1||V_2|)^c+|V_1||V_3|+|V_2||V_3|)$, and determine the existence of a triangle in $G$ within this time.
This contradicts the \UTD{} hypothesis since $(|V_1||V_2|)^c+|V_1||V_3|+|V_2||V_3|= O(n^{2\alpha c}+n^{1+\alpha}) = O(n^{1+\alpha})$, where we used that our assumption $\alpha\leq\frac{1}{2c-1}$ is equivalent to $2\alpha c \le 1+\alpha$.
\end{proof}

\end{exa}

In this section, we showed that if free-connex union extensions capture all UCQs in $\DelayClin$, then the \UTD{} hypothesis holds.
The next section inspects the opposite direction: assuming the \UTD{} hypothesis, we prove the hardness of a large class of UCQs that do not admit free-connex union extensions.

\section{UCQ Classification Based on VUTD}\label{sec:ucq-classification}

In this section, we show the hardness of a large class of UCQs that do not admit a free-connex union extension, assuming the \UTD{} hypothesis.
First, we prove this for unions of a free-connex CQ and a difficult CQ.
Then, we show how \UTD{} can be used instead of hypotheses previously used to show the hardness of UCQs.
Finally, we conclude a dichotomy for a union of two self-join-free CQs.

\subsection{The General Reduction}
\label{sec:gen-reduction}

The following lemma identifies cases in which we can perform a reduction from unbalanced triangle detection to UCQ evaluation.
The reduction requires identifying variable sets in the UCQ that conform to certain conditions. These conditions allow us to encode the tripartite graph in the relations of the query by assigning variables from the same set with the same values.

Given a function $h:X\rightarrow Y$ and a set $S\subseteq Y$, we denote $h^{-1}(S)=\{x \in X \mid h(x) \in S\}$.
Given a CQ $Q$ and sets $X_1,...,X_\ell\subseteq\var(Q)$, we define \[ \connectors_{X_1,...,X_\ell}(Q) = \begin{cases*}
                    \{V\mid R(V)\in\atoms(Q)\}\cup \{\free(Q)\} & if  $\exists X_i : \free(Q)\cap X_i = \emptyset$  \\
                     \{V\mid R(V)\in\atoms(Q)\} & otherwise
                 \end{cases*} \]
When it is clear from the context, we omit the subscript and simply write $\connectors(Q)$.

\begin{lem}[Reduction Lemma]\label{lemma:reduction}
Let $Q=Q_1\cup Q_2$ be non-redundant where $Q_1$ is self-join-free. Suppose that there exist non-empty and disjoint sets $X_1,...,X_\ell\subseteq\var(Q_1)$ with $\ell\geq 3$ such that:
\begin{enumerate}
\item\label{cond:notall}
For every atom $R(V)$ in $Q_1$, there exists $X_i$ s.t. $V\cap X_i = \emptyset$.
\item\label{cond:connected}
$\graph{Q_1}[X_i]$ is connected for all $i$.
\item\label{cond:connectors}
For every $S\in\{\{1,2\},\{1,3,\ldots,\ell\},\{2,3,\ldots,\ell\}\}$, there exists $V\in\connectors(Q_1)$ s.t. $V \cap X_i\neq\emptyset$ for all $i\in S$.
\item\label{cond:otherquery}
For every body-homomorphism $h$ from $Q_2$ to $Q_1$, if we have that $\free(Q_2)\cap h^{-1}(X_\ell) \neq \emptyset$, then $\size{\free(Q_2)\cap h^{-1}(X_\ell)}=1$ and $\size{\free(Q_2)\cap h^{-1}(\bigcup_{1\leq i \leq \ell-1}{X_i})}\leq \ell - 2$.
\end{enumerate}
Then $Q\not\in\DelayClin$, assuming the \UTD{} hypothesis.
\end{lem}

This lemma provides a uniform way of proving the (conditional) hardness of many UCQs. Each set $X_i$ can be seen as giving a role to its variables. In the simple case of $\ell=3$, the three roles correspond to the three vertex sets of the tripartite graph, where $X_3$ is the large set.
As an example, this lemma can be used for Example~\ref{example:separated} by setting $X_1=\{x\}$, $X_2=\{y\}$, and $X_3=\{z\}$. 
In case $\ell>3$, the information regarding the large vertex set is split between $X_3,\ldots,X_\ell$.
The first three conditions of the lemma ensure that we can construct the relations of $Q_1$ so that it detects triangles in the graph, as explained next.
The first condition requires that no atom contains variables of all sets, which restricts the size of the relations and allows for efficient construction. More precisely, it ensures that we do not need to compute all of the triangles already at preprocessing.
The second condition requires that each set is connected, which ensures variables from the same set are assigned the same value in every answer, thus allowing us to give variables of the same set the same role.
Note that the second condition trivially holds when $\size{X_i}=1$.
The third condition ensures that the atoms can encode all three edge sets. The edge sets are usually encoded in the relations, but if we know that the number of answers is restricted, we can also use the free variables as a connector and test whether query answers contain the third edge (as we did in Example~\ref{example:separated}).
The fourth condition restricts the free variables of the other CQ in the union, which ensures that it does not have too many answers, so that the enumeration of the answers of the entire union does not take too long. This easily happens if no free variable of $Q_2$ maps to the large vertex set $X_\ell$, but it can also happen in other cases if we are careful about how many free variables of $Q_2$ map to the vertex sets. 

\begin{proof}[Proof of the Reduction Lemma (\autoref{lemma:reduction})]
We set $\alpha$ to be ${\max\set{{\size{\free(Q_2)}},{\ell-2}}}^{-1}$.
Assume we are given a tripartite graph with vertex sets $V_1$,$V_2$,$V_3$ and edge sets $E_{1,2}$, $E_{2,3}$, $E_{1,3}$ where $|V_1|=|V_2|=\Theta(n^\alpha)$ and $|V_3|=n$.
We set $U_1=V_1$, $U_2=V_2$, and we encode the vertices of $V_3$ as $U_3\times\cdots\times U_\ell$ such that $|U_3|=\ldots=|U_{\ell-1}|=\Theta(n^\alpha)$ and $|U_\ell|=\Theta(n^{1-(\ell-3)\alpha})$. We can do this by putting $\lceil\alpha\log{n}\rceil$ bits of the binary representation of a vertex into each of $U_3,\ldots U_{\ell-1}$ and the remaining bits into $U_\ell$.

We now construct a database instance $D$.
We leave every relation that does not appear in $Q_1$ empty. We next define the other relations according to the atoms of $Q_1$.
Denote by $\calR_{1,2}$ the atoms that contain a variable of $X_1$ and a variable of $X_2$; denote by $\calR_{1,3}$ the atoms that contain at least one variable of $X_i$ for each $i\in\set{1,3,\ldots,\ell}$; and similarly for $\calR_{2,3}$ and $\set{2,3,\ldots,\ell}$. According to condition~\ref{cond:notall}, these sets are disjoint.
We encode the edge sets $E_{1,2},E_{1,3},E_{2,3}$ in the relations that appear in $\calR_{1,2},\calR_{1,3},\calR_{2,3}$, respectively.
Specifically, given an atom $R(\vec{v})$ in $\calR_{1,3}$, for every edge $(v_1,v_3)\in E_{1,3}$, insert a tuple $\tau(\vec{v})$ to $R^D$ as follows: denote by $u_3,\ldots,u_\ell$ the representation of $v_3$ and set $u_1=v_1$; the mapping $\tau$ replaces every variable of the set $X_i$ with the value $u_i$; every variable that does not appear in such a set $X_i$ is replaced with the constant $\bot$.
The construction of relations in $\calR_{2,3}$ proceeds along the same lines.
For atoms in $\calR_{1,2}$ we have a similar construction, except if they contain a variable of $X_i$ for $i>2$, we duplicate each edge and insert it with all possible values in $U_i$. If there are variables of several such sets, we apply all combinations of possible values.
More precisely, consider an atom $R(\vec{v})$ in $\calR_{1,2}$, and let $J$ be the indices $i \in \{3,\ldots,\ell\}$ such that a variable of $X_i$ occurs in $\vec{v}$. For every edge $(u_1,u_2)\in E_{1,2}$ and every value $u_i\in U_i$ for every $i\in J$, insert a tuple $\tau(\vec{v})$ to $R^D$ as follows: the mapping $\tau$ replaces every variable of the set $X_i$ with the value~$u_i$, for each $i \in \{1,2\} \cup J$; every variable that does not appear in such a set $X_i$ is replaced with the constant $\bot$.
Similarly, for atoms that do not belong to the sets $\calR_{1,2},\calR_{1,3},\calR_{2,3}$, we assign variables of $X_i$ with all values of $U_i$ and other variables with $\bot$.
That is, given such an atom $R(\vec{v})$, for every value $u_i\in U_i$ for every $i\in \{1,\ldots,\ell\}$ such that a variable of $X_i$ occurs in $\vec{v}$, insert a tuple $\tau(\vec{v})$ to $R^D$ as follows: the mapping $\tau$ replaces every variable of the set $X_i$ with the value $u_i$ and replaces remaining variables with the constant $\bot$. 
On top of the construction described above, we ensure that each variable has a disjoint domain (e.g., we can concatenate variable names, i.e., if in the above construction we assigned a value $v$ in a position matching a variable $x$ in the corresponding atom, we instead assign the value $\langle v,x\rangle$).
Since no atom contains variables of all $\ell$ sets (Condition~\ref{cond:notall}), each relation size is at most $O(\max\{(n^\alpha)^{\ell-1}, n^{1 - (\ell-3)\alpha} \cdot (n^\alpha)^{\ell-2}\}) \leq O(n^{1+\alpha})$, and so the construction can be done in time $O(n^{1+\alpha})$.
Note that whenever two variables from the same set $X_i$ appear together in the same atom, we assign both with the same value.
Note also that each relation is defined only once since $Q_1$ is self-join-free.

We first claim that the answers to $Q_1$ detect triangles in the graph.
In every answer, for every set $X_i$, all variables of the set are assigned the same value: Variables of the same set that appear in the same atom have the same value by construction, values of $X_i$ that are connected in $\graph{Q_1}$ get the same value by transitivity, and due to Condition~\ref{cond:connected} this applies to all variables of the set.
If we do not use $\free(Q_1)$ as a connector,
Condition~\ref{cond:connectors} ensures that $\calR_{1,2}$, $\calR_{1,3}$, and $\calR_{2,3}$ are non-empty, so the answers are filtered by at least one atom that corresponds to each of the three edge sets, and so answers correspond to triangles. That is, $Q_1$ has an answer if and only if the graph has a triangle.
If we do use $\free(Q_1)$ as a connector, one of the sets $\calR_{1,2}$, $\calR_{1,3}$, and $\calR_{2,3}$ may be empty, and so some edge is not verified. This means that the answers to $Q_1$ are candidates for triangles, and we need to check every answer for the missing edge. In this case, by definition of the connectors set, there exists $i$ such that $\free(Q_1)\cap X_i=\emptyset$ (which is why $\free(Q_1)$ can only contain one of the three sets of Condition~\ref{cond:connectors}, and therefore at most one of the sets $\calR_{1,2}$, $\calR_{1,3}$, and $\calR_{2,3}$ may be empty). If $i=\ell$, the number of answers to $Q_1$ is at most $O(n^{1-(\ell-3)\alpha}(n^\alpha)^{\ell-2}) = O(n^{1+\alpha})$. If $i<\ell$, it is at most $O((n^\alpha)^{\ell-1}) \leq O(n^{1+\alpha})$. Thus, this check that takes constant time for each answer can be done in time $O(n^{1+\alpha})$ in total.

We now show that the answers to $Q_2$ do not interfere with detecting the triangles efficiently.
First note that we can distinguish the answers of $Q_1$ from those of $Q_2$ since we assigned different variables with disjoint domains: Since we assume $Q_1$ is not contained in $Q_2$, any body-homomorphism $h$ from $Q_2$ to $Q_1$ is not a full homomorphism, so $\free(Q_1)\neq h(\free(Q_2))$, and answers of different CQs contain different domains.
We show next that, due to Condition~\ref{cond:otherquery}, $Q_2$ does not have too many answers.
If no free variable of $Q_2$ maps to a variable of $X_\ell$, then the domain of any free variable in $Q_2$ is of size $O(n^\alpha)$. 
Since we defined $\alpha$ such that $\size{\free(Q_2)} \le 1/\alpha$, $Q_2$ has at most $n$ answers.
Otherwise, exactly one free variable of $Q_2$ maps to a variable of $X_\ell$, and at most $\ell-2$ free variables of $Q_2$ map to variables of the other sets. In this case, the number of answers to $Q_2$ is $O(n^{1-(\ell-3)\alpha}(n^\alpha)^{\ell-2}) = O(n^{1+\alpha})$.

Using this construction, enumerating $O(n^{1+\alpha})$ answers to the union allows determining whether the given graph contains a triangle. If $\free(Q_1)$ is not used as a connector, there is a triangle if and only if $Q_1$ has answers. After enumerating $O(n^{1+\alpha})$ answers to the union, we will either get an answer to $Q_1$ or the enumeration will terminate signaling that there are no such answers. If $\free(Q_1)$ is used as a connector, the union has $O(n^{1+\alpha})$ answers. We can enumerate all answers and test for every answer to $Q_1$ whether it corresponds to a triangle. This way, if $Q\in\DelayClin$, we detect triangles in the given graph in time $O(n^{1+\alpha})$, contradicting the \UTD{} hypothesis.
\end{proof}

\subsection{A Non-Provided Difficult Structure}\label{sec:lower-bound}

We want to show that we can use this reduction to show the hardness of some UCQs that contain one free-connex CQ and one difficult CQ.
The difficult CQ is self-join-free, and we first notice that there is at most one body-homomorphism mapping to a self-join-free CQ. 

\begin{prop}\label{prop:one-homo}
Let $Q=Q_1\cup Q_2$ where $Q_1$ is self-join-free.
There is at most one body homomorphism from $Q_2$ to $Q_1$.
\end{prop}
\begin{proof}
Consider body-homomorphisms
$h_1$ and $h_2$ from $Q_2$ to $Q_1$.
If $h_1\neq h_2$, there exists a variable $v\in\var(Q_2)$ such that $h_1(v)\neq h_2(v)$.
Consider an atom $R(\vec{v})$ in $Q_2$ such that $v\in\vec{v}$.
Since they are body homomorphisms, $R(h_1(\vec{v}))$ and $R(h_2(\vec{v}))$ are in $Q_1$.
This contradicts the fact that $Q_1$ is self-join-free.
\end{proof}

We show that the reduction from the Reduction Lemma can be applied whenever the free-connex CQ does not provide all variables of some difficult structure in the difficult CQ.

\begin{lem}\label{lemma:original-apply}
Consider a UCQ $Q=Q_1\cup Q_2$ where $Q_1$ is difficult and $Q_2$ is free-connex.
If there exists a difficult structure in $Q_1$ that is not provided by $Q_2$, then the Reduction Lemma can be applied, and thus $Q\not\in\DelayClin$ assuming the \UTD{} hypothesis.
\end{lem}

\begin{proof}
We separate to cases according to the type of difficult structure.
In all cases, we show how to select sets $X_i$ such that the first three conditions hold and $X_\ell$ consists of a single unprovided variable $v$.
Since $v$ is not provided and $Q_2$ is free-connex, either there is no body-homomorphism $h$ from $Q_2$ to $Q_1$, or $v\not\in h(\free(Q_2))$.
In both cases, Condition~\ref{cond:otherquery} holds.

In case of a tetra, denote its variables by $\{x_1,\ldots,x_k\}$ such that $x_k$ is not provided, and set $X_i=\{x_i\}$ for $1 \le i \le k$, that is, $\ell = k$.
Since no edge contains all tetra variables, Condition~\ref{cond:notall} holds.
Condition~\ref{cond:connected} trivially holds since the sets $X_i$ are of size one. Condition~\ref{cond:connectors} holds since the tetra hyperedges form the connectors of $\{x_1,x_2\}$, $\{x_2,\ldots,x_k\}$, and $\{x_1,x_3,\ldots,x_k\}$.

In case of a chordless cycle, denote it as $x_1,\ldots,x_k,x_1$ such that $x_k$ is not provided. Set $X_1=\{x_1,..,x_{k-2}\}$, $X_2=\{x_{k-1}\}$, and $X_3=\{x_k\}$, that is, $\ell = 3$.
As the cycle is chordless, Condition~\ref{cond:notall} holds.
Condition~\ref{cond:connected} holds due to the path $x_1,..,x_{k-2}$ that lies on the cycle. Condition~\ref{cond:connectors} holds due to the three hyperedges containing $\{x_{k-2},x_{k-1}\}$, $\{x_{k-1},x_{k}\}$ and $\{x_{k},x_{1}\}$ on the cycle.

In the case of a free-path, we split into two cases.
If an end variable of the path is not provided,
denote the path by $x,z_1,\ldots,z_k,y$ such that $y$ is not provided. We set $X_1=\{x\}$, $X_2=\{z_1,\ldots,z_k\}$ and $X_3=\{y\}$, that is, $\ell = 3$.
Otherwise, if both end variables are provided, a middle variable is not provided. Denote this variable by $z$, and the path by $x_1,\ldots,x_k,z,y_1,\ldots,y_m$. We set $X_1=\{x_1,\ldots,x_k\}$, $X_2=\{y_1,\ldots,y_m\}$ and $X_3=\{z\}$.
In both cases,
Condition~\ref{cond:notall} holds since the path is chordless and so no atom contains both a variable with $x$ in the name and a variable with $y$.
Condition~\ref{cond:connected} holds due to the relevant subpaths. For Condition~\ref{cond:connectors}, the connection between the sets containing the end variables is done through the connector $\free(Q_1)$; we can use this connector since the interior of the path holds no free variables, so $X_3\cap\free(Q_1)=\emptyset$. The other two connectors appear on the path.
\end{proof}

Note that the Reduction Lemma and \autoref{lemma:original-apply} do not require $Q_2$ to be self-join-free. Consider as an example the following modification of \autoref{example:separated} with
$Q_1(x,y,w)\cqa R_1(x,z),R_2(z,y),R_3(y,w)$ and $Q_2(x,y,w)\cqa R_1(x,t_1),R_3(y,t_2),R_3(w,t_3)$.
The reduction can be applied here with $X_1=\{x\}$, $X_2=\{y\}$, and $X_3=\{z\}$, and we conclude that this query is not in $\DelayClin$ assuming the \UTD{} hypothesis.

\subsection{Completeness for Binary Relations}\label{sec:completeness}

Following the previous section, it is left to handle the case that all difficult structures in $Q_1$ are provided by $Q_2$.
For exposition purposes, we start with the case where the difficult CQ contains only binary relations, and we show that, in this case, if the UCQ is not covered by \autoref{lemma:original-apply}, then the query necessarily has a union extension. We then conclude that, when considering such UCQs with binary relations, union extensions exactly capture $\DelayClin$, assuming the \UTD{} hypothesis. Since the main result of this section is generalized in the next section to relations of all arities, the proofs are excluded from the paper body, and they appear in \autoref{sec:binary-proof}.

Recall that if a UCQ has a free-connex union extension, then it is in $\DelayClin$.
We define a process of generating a union extension of a difficult CQ by repeatedly adding virtual atoms that correspond to difficult structures (thus eliminating the difficult structures).

\begin{defi}\label{def:resolution}
Let $Q = Q_1 \cup Q_2$ be a union of a difficult CQ $Q_1$ and a free-connex CQ $Q_2$. We define a \emph{resolution step} over $Q$: if there is a difficult structure in $Q_1$ with variables~$V$, and $V$ is provided by $Q_2$, extend $Q_1$ with a new atom with the variables $V$. \emph{Resolving} the UCQ $Q$ means applying resolution steps to $Q_1$ until it is no longer possible. We denote the resulting UCQ by $Q^+$, and we say that $Q^+$ is \emph{resolved}.
\end{defi}

Note that resolution describes a special case of union extensions,
and that several resolution steps may be required in case $Q_1$ contains several difficult structures or in case a resolution step introduces a new difficult structure.
We can show that for binary relations, when all variables that participate in difficult structures are provided, the resolution process given in \autoref{def:resolution} results in a free-connex CQ.

\begin{lem}\label{lemma:binary-free-connex}
Let $Q=Q_1\cup Q_2$ where $Q_1$ is difficult and consists of binary atoms, $Q_2$ is free-connex,
and $Q_2$ provides all difficult structures in $Q_1$.
Then the resolved $Q_1^+$ is free-connex.
\end{lem}

By combining \autoref{lemma:binary-free-connex} with \autoref{lemma:original-apply} and the Reduction Lemma, we get that free-connex union extensions capture all UCQs in $\DelayClin$ that contain one free-connex CQ and one difficult CQ when the relations are binary.

\begin{thm}\label{thm:binary-hardness}
Let $Q=Q_1\cup Q_2$ be a non-redundant UCQ where $Q_1$ is difficult and comprises of binary atoms and $Q_2$ is free-connex.
Assuming the \UTD{} hypothesis, the following are equivalent:
\begin{itemize}
    \item $Q\in\DelayClin$.
    \item $Q$ admits a free-connex union extension.
    \item The resolution process of Definition~\ref{def:resolution} turns $Q_1$ free-connex. 
\end{itemize}
\end{thm}

\subsection{Completeness for General Arity}
\label{sec:general-arity}

\autoref{lemma:binary-free-connex} no longer holds when we allow general arities. The current proof fails since, unlike for graphs, the existence of a simple cycle in a hypergraph does not imply the existence of a chordless cycle.
However, this does not mean that \autoref{thm:binary-hardness} does not hold for general arity or that our techniques cannot be used in this case.
Here is an example (left open by prior work) for when the generalization of \autoref{lemma:binary-free-connex} for general arity does not hold, but we can still use the Reduction Lemma to show that the UCQ is hard assuming the \UTD{} hypothesis.

\begin{exaC}[{\cite[Example 45]{DBLP:conf/pods/CarmeliK19}}]
Let $Q=Q_1\cup Q_2$ with:
\begin{align*}
	Q_1(x_2,\ldots,x_k)\cqa& \{R_i (x_1,...,x_{i-1},x_{i+1},...,x_k)\mid 1\leq i\leq k-1\}\\
	Q_2(x_2,\ldots,x_k)\cqa& R_1(x_2,\ldots,x_{k-1},x_1),R_2(x_k,x_3,\ldots,x_{k-1},v).
\end{align*}
If $k\geq 4$, the query $Q_1$ is cyclic and $Q_2$ is free-connex. Although $Q_2$ provides the tetra $\{x_1,\ldots,x_{k-1}\}$, adding a virtual atom with these variables does not result in a free-connex extension, as this extension is exactly a tetra $\{x_1,\ldots,x_{k}\}$.
The Reduction Lemma can be applied here by setting $X_i=\{x_i\}$:
Condition~\ref{cond:notall} holds since no edge contains $\{x_1,\ldots,x_k\}$;
Condition~\ref{cond:connected} holds trivially since the sets are of size one;
Condition~\ref{cond:connectors} holds due to the hyperedges
$\{x_1,\ldots,x_k\}\setminus\{x_1\}$,
$\{x_1,\ldots,x_k\}\setminus\{x_2\}$, and $\{x_1,\ldots,x_k\}\setminus\{x_3\}$; and
Condition~\ref{cond:otherquery} holds since $x_k\not\in h(\free(Q_2))$, where $h$ is the unique homomorphism from $Q_2$ to $Q_1$.
Thus, assuming the \UTD{} hypothesis, $Q\not\in\DelayClin$.
\qed
\end{exaC}

We next generalize \autoref{thm:binary-hardness} to general arity using the Reduction Lemma. As the full proof is involved, it is deferred to \autoref{sec:general-arity-proof}.

\begin{lem}\label{lemma:general-arity-reduction-applies}
Let $Q=Q_1\cup Q_2$ be a non-redundant UCQ where $Q_1$ is difficult and $Q_2$ is free-connex.
If the resolved $Q_1^+$ is not free-connex,
then the Reduction Lemma can be applied, and thus $Q\not\in\DelayClin$ assuming the \UTD{} hypothesis.
\end{lem}

\begin{proof}[Proof Sketch]
Since $Q_1$ is difficult, it is self-join-free, and so there is at most one body-homomorphism from $Q_2$ to $Q_1$.
If no such homomorphism exists, then in particular, $Q_2$ does not provide any difficult structure in $Q_1$, and according to \autoref{lemma:original-apply}, the Reduction Lemma can be applied.
Now we can assume that there is one such body-homomorphism $h$.

Let $Q^+$ be the fully resolved $Q$ according to \autoref{def:resolution}.
Since $Q_1^+$ is not free-connex,
there is a difficult structure in $Q_1^+$, and since $Q_1$ is fully resolved, $h(\free(Q_2))$ does not contain all of the variables in this structure. We prove that the reduction can be applied in this case by induction on the extension steps.
Specifically, we prove the following claim by induction on $t$: if there exists a variable $v\not\in h(\free(Q_2))$ in a difficult structure in an extension of $Q_1$ obtained after $t$ resolution steps, then the Reduction Lemma can be applied.

The base case is given by \autoref{lemma:original-apply}: If $Q_2$ does not provide some difficult structure in $Q_1$, then the Reduction Lemma can be applied.
We next show the induction step.
Assume there exists a variable $v\not\in h(\free(Q_2))$ in a difficult structure $S$ in an extension of $Q_1$. If all edges in this structure appear in $Q_1$, then by \autoref{lemma:original-apply}, the Reduction Lemma can be applied.
Otherwise, take the last extension $Q_1'$ in the sequence of resolution steps where this structure does not appear. This means that $Q_1'$ contains all edges of $S$ except one, and this missing edge comprises the nodes of some difficult structure in $Q_1'$, where all these nodes are in $h(\free(Q_2))$.
Note that if we show that $v$ is in a difficult structure in $Q_1'$, we can use the induction assumption to show that the Reduction Lemma applies.
\autoref{sec:general-arity-proof} contains a rigorous case distinction: we first separate according to the type of difficult structure of $S$, and then by the type of difficult structure in $Q_1'$ that causes the addition of the last edge of $S$.
To show the induction step, we sometimes use the induction assumption and sometimes directly identify structures in $Q_1$ on which we can apply our reduction.
\end{proof}

By combining \autoref{lemma:general-arity-reduction-applies} with the Reduction Lemma, we get that free-connex union extensions capture all UCQs in $\DelayClin$ that contain one free-connex CQ and one difficult CQ. Moreover, such union extensions can always be obtained using the resolution process of \autoref{def:resolution}.

\begin{thm}\label{thm:general-hardness}
Let $Q=Q_1\cup Q_2$ be a non-redundant UCQ where $Q_1$ is difficult and $Q_2$ is free-connex.
Assuming the \UTD{} hypothesis, the following are equivalent:
\begin{itemize}
    \item $Q\in\DelayClin$.
    \item $Q$ admits a free-connex union extension.
    \item The resolution process of Definition~\ref{def:resolution} turns $Q_1$ free-connex. 
\end{itemize}
\end{thm}

\begin{proof}
If the resolution process of Definition~\ref{def:resolution} turns $Q_1$ free-connex, the resolved UCQ is a free-connex union extension of $Q$, and according to \autoref{theorem:UCQtractability}, $Q\in\DelayClin$.
If $Q\in\DelayClin$, the resolved UCQ is free-connex, according to \autoref{lemma:general-arity-reduction-applies} and assuming the \UTD{} hypothesis.
\end{proof}

\subsection{Proof of \autoref{lemma:general-arity-reduction-applies}}
\label{sec:general-arity-proof}
\pagebreak
\subsubsection*{Preparations}

We first name some of the structures that will be used in the proof.

\begin{defi}
Consider a hypergraph describing a CQ.
\begin{itemize}
    \item Nodes $v,u_1,\ldots,u_k$ form a \emph{hand-fan (from $u_1$ to $u_k$ centered in~$v$)} if (1) $u_1,\ldots,u_k$ is a chordless path with $k\ge 3$, (2)~for every $1 \le i < k$ the hypergraph has an edge containing $\{v,u_i,u_{i+1}\}$, and (3) $v \ne u_i$ for all $i$.
    \item A \emph{free-hand-fan} is a hand-fan where $u_1,\ldots,u_k$ is a free-path.
    \item Nodes $v,u_1,\ldots,u_k$ form a \emph{flower (centered in $v$)} if (1) $u_1,\ldots,u_k$ is a chordless cycle with $k\ge 3$, and (2) for every $1 \le i < k$  the hypergraph has an edge containing $\{v,u_i,u_{i+1}\}$, and the hypergraph has an edge containing $\{v,u_k,u_1\}$.
\end{itemize}
\end{defi}

We define the \emph{chordless shortening} of a path $v_1\ldots,v_\ell$ to be the chordless path obtained by the following process: start with $i=1$, let $j$ be the largest index such that $v_j$ is neighbor of $v_i$, mark all vertices $v_k$ with $i<k<j$, and repeat with $j$ instead of $i$ until reaching the end of the path. Finally, remove all marked nodes.

\begin{lem}[Implications of a Simple Cycle]\label{lemma:simple-cycle-implications}
If a node $v$ is in a simple cycle $v,v_2\ldots,v_\ell,v$, then one of the following holds:
\begin{itemize}
    \item $v$ appears in a chordless cycle (possibly a triangle).
    \item There is an edge containing $v$ and its two cycle neighbors $v_2,v_\ell$.
    \item There is a hand-fan from $v_2$ to $v_k$ centered in $v$, given by the chordless shortening of the path $v_2,\ldots,v_\ell$.
\end{itemize}
\end{lem}

\begin{proof}
Denote by $v_2=u_1,\ldots,u_k=v_\ell$ the chordless shortening of the path $v_2,\ldots,v_\ell$.
The first case is that some $u_i$ is not a neighbor of $v$.
Let $u_s$ be the last vertex before $u_i$ which is a neighbor of~$v$, and let $u_t$ be the first vertex after $u_i$ which is a neighbor of $v$. As we took a chordless path, $u_s$ and $u_t$ are not neighbors, and so we know there is no edge containing $\{v,u_s,u_t\}$.
Thus, the cycle $v,u_s,\ldots,u_t,v$ is chordless.
The second case is that all variables on the cycle $v,u_1,\ldots,u_k,v$ are neighbors of $v$. If there exists $1\leq j <k$ such that there is no edge containing $\{v,u_j,u_{j+1}\}$, then $v,u_j,u_{j+1}$ form a chordless cycle of length three.
Otherwise, for any $1 \le j < k$ there is an edge containing $\{v,u_j,u_{j+1}\}$. Then, if $k=2$, there is an edge containing $\{v,u_1,u_2\}$ (that is, $v$ and its two cycle neighbors), and if $k\ge 3$, nodes $v,u_1,\ldots,u_k$ form a hand-fan.
\end{proof}

We show that the Reduction Lemma can be applied in the following three cases.

\begin{lem}[Reduction for Free-Hand-Fan]\label{lemma:free-hand-fan-reduction}
Let $Q=Q_1\cup Q_2$ be a non-redundant UCQ where $Q_1$ is self-join-free and $h$ is a body-homomorphism from $Q_2$ to $Q_1$.
If $v\not\in h(\free(Q_2))$ is the center of a free-hand-fan in $Q_1$, then the Reduction Lemma can be applied, and thus $Q\not\in\DelayClin$ assuming the \UTD{} hypothesis.
\end{lem}
\begin{proof}
Denote the free-hand-fan nodes by $v,u_1,\ldots,u_k$.
Set $X_1=\{u_1\}$, $X_2=\{u_k\}$, $X_3=\{u_2,\ldots,u_{k-1}\}$ and $X_4=\{v\}$.
Since $u_1,\ldots,u_k$ is a chordless path, no edge contains $4$ of the hand-fan variables, and so
Condition~\ref{cond:notall} holds.
Condition~\ref{cond:connected} holds since $u_2,\ldots,u_{k-1}$ is a path.
Condition~\ref{cond:connectors} holds as
$\{u_{k-1},u_k,v\}$ connects $\{X_2,X_3,X_4\}$, 
$\{u_1,u_2,v\}$ connects $\{X_1,X_3,X_4\}$, and
the free variables connect $\{X_1,X_2\}$, as $u_1,u_2$ are free in $Q_1$ and none of the variables in $X_3$ are free.
Condition~\ref{cond:otherquery} holds since $v\not\in h(\free(Q_2))$.
\end{proof}

\begin{lem}[Reduction for Flower]\label{lemma:flower-reduction}
Let $Q=Q_1\cup Q_2$ be a non-redundant UCQ where $Q_1$ is self-join-free and $h$ is a body-homomorphism from $Q_2$ to $Q_1$.
If $v\not\in h(\free(Q_2))$ is the center of a flower in $Q_1$, then the Reduction Lemma can be applied, and thus $Q\not\in\DelayClin$ assuming the \UTD{} hypothesis.
\end{lem}
\begin{proof}
Denote the flower nodes by $v,u_1,\ldots,u_k$.
Set $X_1=\{u_1\}$, $X_2=\{u_2\}$, $X_3=\{u_3,\ldots,u_k\}$ and $X_4=\{v\}$.
Condition~\ref{cond:notall} holds since no edge contains $4$ of the flower nodes, as $u_1,\ldots,u_k$ is a chordless cycle.
Condition~\ref{cond:connected} holds since $u_3,\ldots,u_k$ is a path.
Condition~\ref{cond:connectors} holds as
$\{u_1,u_2,v\}$ connects $\{X_1,X_2\}$,
$\{u_2,u_3,v\}$ connects $\{X_2,X_3,X_4\}$, and
$\{u_1,u_k,v\}$ connects $\{X_1,X_3,X_4\}$.
Condition~\ref{cond:otherquery} holds since $v\not\in h(\free(Q_2))$.
\end{proof}

\begin{lem}[Reduction for Almost Tetra]\label{lemma:almost-tetra-reduction}
Let $Q=Q_1\cup Q_2$ be a non-redundant UCQ where $Q_1$ is self-join-free and $h$ is a body-homomorphism from $Q_2$ to $Q_1$.
If there are variables $x_1,\ldots x_k$ with $k\geq 4$ such that $x_k\not\in h(\free(Q_2))$ and $Q_1$ has an edge containing $\{x_1,\ldots,x_k\}\setminus\{x_i\}$ for every $1\leq i\leq k-1$, but no edge containing all of $\{x_1,\ldots,x_k\}$,
then the Reduction Lemma can be applied, and thus $Q\not\in\DelayClin$ assuming the \UTD{} hypothesis.
\end{lem}
\begin{proof}
Set $X_i=\{x_i\}$.
Condition~\ref{cond:notall} holds since no edge contains $\{x_1,\ldots,x_k\}$.
Condition~\ref{cond:connected} holds trivially since the sets are of size one.
Condition~\ref{cond:connectors} holds due to the edges
$\{x_1,\ldots,x_k\}\setminus\{x_1\}$,
$\{x_1,\ldots,x_k\}\setminus\{x_2\}$, and $\{x_1,\ldots,x_k\}\setminus\{x_3\}$.
Condition~\ref{cond:otherquery} holds since $x_k\not\in h(\free(Q_2))$.
\end{proof}

\subsubsection*{Proof Setup}

Let $Q = Q_1 \cup Q_2$ be non-redundant where $Q_1$ is difficult and $Q_2$ is free-connex, and assume that the resolved $Q_1^+$ is not free-connex. We want to show that the Reduction Lemma can be applied, to prove \autoref{lemma:general-arity-reduction-applies}. 
Since $Q_1$ is difficult, it is self-join-free, and so there is at most one body-homomorphism from $Q_2$ to $Q_1$.
If no such homomorphism exists, then in particular, $Q_2$ does not provide any difficult structure in $Q_1$, and according to \autoref{lemma:original-apply}, the Reduction Lemma can be applied.
So we can assume that there is one such body-homomorphism $h$.
Since the resolved $Q_1^+$ is not free-connex, there is a difficult structure in $Q_1^+$, and since $Q_1$ is fully resolved, $h(\free(Q_2))$ does not contain all of the variables in this structure. We prove that the reduction can be applied in this case by induction on the extension steps.
Specifically, we prove the following claim by induction on $t$: \emph{if there exists a variable $v\not\in h(\free(Q_2))$ in a difficult structure in an extension of $Q_1$ obtained after $t$ resolution steps, then the Reduction Lemma can be applied}.
Since the precondition of this claim is satisfied for $Q_1^+$, \autoref{lemma:general-arity-reduction-applies} follows after proving this claim.

The base case of the induction is given by \autoref{lemma:original-apply}: If $Q_2$ does not provide some difficult structure in $Q_1$, then the Reduction Lemma can be applied.
It remains to show the induction step.

\subsubsection*{Induction Step}

We can now prove the induction step for our claim from above.
Assume there exists a variable $v\not\in h(\free(Q_2))$ in a difficult structure $S$ in an extension of $Q_1$ after $t$ resolution steps. If all edges in this structure appear in $Q_1$, then by \autoref{lemma:original-apply}, the Reduction Lemma can be applied.
Otherwise, take the last extension $Q_1'$ in the sequence of resolution steps where this structure does not appear. This means that $Q_1'$ contains all edges of $S$ except one, and this missing edge comprises of the nodes of some difficult structure in $Q_1'$, where all these nodes are in $h(\free(Q_2))$.
Note that if we show that $v$ is in a difficult structure in $Q_1'$, we can use the induction hypothesis to show that the Reduction Lemma applies.
We now embark on a rigorous case distinction, and we first distinguish cases according to the type of difficult structure of $S$.

\subsubsection{Tetra}
The first case is that $S$ is a tetra of size $k\ge 4$. In this case, $S$ is introduced to $Q_1'$ by adding an edge containing $k-1$ of its variables, all of them in $h(\free(Q_2))$. Denote this edge by $\{x_1,...,x_{k-1}\}$. As $v$ is part of $S$ and $v\not\in h(\free(Q_2))$, we conclude that $v$ is the remaining variable of the tetra and that no edge in $Q_1'$ (or $Q_1$) contains $\{x_1,\ldots,x_{k-1},v\}$. Since all other tetra edges contain $v$ and $v\not\in h(\free(Q_2))$, we know that all other tetra edges already appear in $Q_1$ as they cannot be added as part of an extension. We can use the Reduction Lemma in this case according to \autoref{lemma:almost-tetra-reduction}.

\subsubsection{Free-Path}
Consider the case that $S$ is a free-path $u_1,\ldots,u_k$. Let $u_j,u_{j+1}$ be the free-path edge that is missing in $Q_1'$. This means that $u_j,u_{j+1}\in h(\free(Q_2))$, so $v\not\in\{u_j,u_{j+1}\}$, and we can assume without loss of generality that $v=u_i$ with $i<j$. 
Denote by $S'$ the difficult structure in $Q_1'$ that causes the addition of the edge $\{u_j,u_{j+1}\}$.
Observe that $S$ and $S'$ intersect exactly in the nodes $\{u_j,u_{j+1}\}$: If $S'$ would contain another node $u_x$ of $S$, then $u_x$ would also appear in the edge with $\{u_j,u_{j+1}\}$ added in the next extension step, contradicting the free-path $S$ being chordless.
We now distinguish cases according to the type of difficult structure of $S'$.
Note that, as covering tetras does not connect pairs of variables that were not neighbors before, this structure cannot be a tetra.

\paragraph{Covering a Cycle}
In this case, $u_j$ and $u_{j+1}$ appear together in a chordless cycle $S'$ in $Q_1'$. Denote the two paths between $u_j$ and $u_{j+1}$ on the cycle $S'$ by $u_j=t_1,\ldots,t_n=u_{j+1}$ and  $u_j=b_1,\ldots,b_m=u_{j+1}$.
We now distinguish the following cases:
\begin{itemize}
\item
If $u_i$ has neighbors in both $t_2,\ldots,t_n$ and $b_2,\ldots,b_m$, let $t_p$ and $b_q$ be its neighbors with largest indices $p$ and $q$.
Since $S$ is chordless, $u_{j+1}$ is not a neighbor of $u_i$, and so $p<n$ and $q<m$.
Then, $u_i,t_p\ldots,t_n=b_m,\ldots,b_q,u_i$ is a chordless cycle of length at least $4$.
\item
In the remaining case, we can assume without loss of generality that $t_2,\ldots,t_{n}$ are not neighbors of $u_i$. We distinguish:
    \begin{itemize}
    \item
    If there is an edge from a node before $u_i$ to $t_2,\ldots,t_{n}$, pick $p<i$ maximal such that $u_p$ has an edge to $t_2,\ldots,t_{n}$, and pick $1<q\le n$ minimal such that $u_p-t_q$ is an edge. Let $P$ be the chordless shortening of the path $u_p,\ldots,u_j=t_1,\ldots,t_q$. This path $P$ starts with $u_p,\ldots,u_i,u_{i+1}$ (since $S$ is a free-path and thus chordless, by maximality of $p$, and since $u_i$ has no neighbors in $t_2,\ldots,t_{n}$). It follows that $P$ together with the edge $u_p-t_q$ forms a chordless cycle of length at least 4 containing $v=u_i$. 
    \item
    If there is no edge from any of $u_1,\ldots,u_{i-1}$ to any of $t_2,\ldots,t_{n-1}$, let $P$ be the chordless shortening of the path $u_1,\ldots,u_j=t_1,\ldots,t_n=u_{j+1},\ldots,u_k$. 
    This path $P$ starts with $u_1,\ldots,u_i,u_{i+1}$ as the first $i$ variables do not have chords on the path. 
    The path $P$ starts and ends in a free variable, and $u_2,\ldots,u_{i+1}$ are not free, so by taking the prefix of $P$ that stops at the second free variable along $P$, we obtain a free-path containing $u_i$.
    \end{itemize}
\end{itemize}

In each case, we can apply the induction hypothesis to prove the inductive step.

\paragraph{Covering a Free-Path}
In this case, $u_j$ and $u_{j+1}$ appear together in a free-path $S'$ in $Q_1'$. Denote this free-path by $f_1,\ldots,f_n$ such that $u_j=f_a$ and $u_{j+1}=f_b$ for some $1\le a < b \le n$.
We distinguish the following cases:
\begin{itemize}
\item
If there is no edge from any node in $u_1,\ldots,u_{i-1}$ to any node in $f_{a+1}\ldots,f_{b-1}$, consider the path $u_1,\ldots,u_j=f_a,\ldots,f_b=u_{j+1},\ldots,u_k$.
The chordless shortening of this path starts with $u_1,\ldots,u_{i}$, since the first $i-1$ nodes do not have chords with any of the path nodes, and it consists of at least $3$ nodes because $u_1$ and $u_k$ are not neighbors, so it is a free-path containing $u_i$. Thus, we can apply the induction hypothesis and are done.
\item
If there is no edge from any node in $u_1,\ldots,u_{i-1}$ to any node in $f_1,\ldots,f_{a-1}$, consider the path $u_1,\ldots,u_j=f_a,\ldots,f_1$.
The chordless shortening of this path starts with $u_1,\ldots,u_{i}$ as the first $i-1$ nodes do not have chords with any of the path nodes, so if it consists of at least $3$ nodes, then it is a free-path containing $u_i$. The remaining case is that the chordless shortening has length 2, which can only happen if $i=1$, and the chordless shortening is $u_1,f_1$. As $i=1$, the previous case applies and shows that $u_i$ is part of a free-path. In both cases, we can apply the induction hypothesis and are done.
\item
The last case is that there is an edge from some node in $u_1,\ldots,u_{i-1}$ to some node in $f_1,\ldots,f_{a-1}$, and there is an edge from some node in $u_1,\ldots,u_{i-1}$ to some node in $f_{a+1}\ldots,f_{b-1}$.
Denote by $u_p-f_q$ an edge with $1\le p<i$ and $1\le q<a$.
Consider the cycle $u_p,\ldots,u_j=f_a,\ldots,f_q,u_p$.
This is a simple cycle because $S$ and $S'$ intersect exactly in the nodes $\{u_j,u_{j+1}\}$, as we have previously argued.
Since $S$ is chordless, we also conclude that $\{u_{i-1},u_i,u_{i+1}\}$ do not appear together in an edge.
According to \autoref{lemma:simple-cycle-implications}, either $u_i$ is part of a chordless cycle in $Q_1'$ (so we can apply the induction assumption and we are done) or it is a center of a hand-fan from $u_{i-1}$ to $u_{i+1}$.
Note that, as there cannot be edges containing $\{u_i,u_{i+1},u_{i+2}\}$ or $\{u_{i-2},u_{i-1},u_i\}$ by the chordlessness of $S$, the hand-fan path uses as intermediate nodes only nodes of $f_1,\ldots,f_{a-1}$.
By applying the same argument on $f_{a+1},\ldots,f_{b-1}$, we obtain that $u_i$ appears in a chordless cycle in $Q_1'$ or $u_i$ is the center of a hand-fan from $u_{i-1}$ to $u_{i+1}$ through nodes of $f_{a+1},\ldots,f_{b-1}$.
By assembling the two hand-fans we discovered, we obtain that $u_i$ is the center of a flower, and according to \autoref{lemma:flower-reduction} the Reduction Lemma can be applied.
\end{itemize}

\subsubsection{Cycle}
It remains to consider the case that $S$ is a chordless cycle.
Denote this cycle by $u_1,\ldots,u_k$ such that the edge containing $\{u_1,u_k\}$ does not appear in $Q_1'$. Note that $v=u_i$ for some $1<i<k$, as $v$ cannot be part of an extension edge.
Denote by $S'$ the difficult structure in $Q_1'$ that causes the addition of the edge $\{u_1,u_k\}$.
As before, observe that $S$ and $S'$ intersect exactly in the nodes $\{u_1,u_k\}$: If $S'$ would contain another node $u_x$ of $S$, then $u_x$ would also appear in the edge with $\{u_1,u_k\}$ added in the next extension step, contradicting the cycle $S$ being chordless.
We now further distinguish cases according to the type of difficult structure of $S'$.
Note that, as covering tetras does not connect pairs of nodes that were not neighbors before, this structure cannot be a tetra.

\paragraph{Covering a Cycle}
In this case, $u_1$ and $u_{k}$ appear together in a chordless cycle $S'$ in $Q_1'$.
Denote by $P_1$ and $P_2$ the two paths in $Q_1'$ remaining from the covered cycle when removing $u_1$ and $u_k$. 
Consider the cycle obtained by concatenating $P_1$ with $u_1,u_2,\ldots,u_k$. This is a simple cycle since $S$ and $S'$ intersect exactly in $\{u_1,u_k\}$.
Since $S$ is chordless, we also conclude that $\{u_{i-1},u_i,u_{i+1}\}$ do not appear together in an edge.
According to \autoref{lemma:simple-cycle-implications}, either $u_i$ is part of a chordless cycle in $Q_1'$ (so we can apply the induction assumption and we are done) or it is a center of a hand-fan from $u_{i-1}$ to $u_{i+1}$.
Note that, as there cannot be edges containing $\{u_i,u_{i+1},u_{i+2}\}$ or $\{u_{i-2},u_{i-1},u_i\}$  (in case $u_{i-2}$ and $u_{i+1}$ are defined) by the chordlessness of $S$, the hand-fan path uses as intermediate nodes only nodes of $P_1$.
By applying the same argument on $P_2$, we get that $v=u_i$ appears in a chordless cycle in $Q_1'$ (so we are done) or $u_i$ is the center of a hand-fan from $u_{i-1}$ to $u_{i+1}$ through nodes of $P_2$.
By assembling the two hand-fans we discovered, we obtain that $u_i$ is the center of a flower, and according to \autoref{lemma:flower-reduction} the Reduction Lemma can be applied.

\paragraph{Covering a Free-Path}
In this case, $u_1$ and $u_k$ appear together in a free-path $S'$ in $Q_1'$. Denote this free-path by $f_1,\ldots,f_n$ such that $u_1=f_a$ and $u_k=f_b$ for some $1\le a < b \le n$.
We distinguish the following cases:

\begin{itemize}
\item
If no edge exists from $v=u_i$ to any node in $f_{a+1},\ldots,f_{b-1}$,
consider the cycle $u_1,\ldots,u_k=f_b,\ldots,f_a=u_1$.
Note that $v$ is part of this cycle, but it is not part of any chords on this cycle.
This is a simple cycle because the intersection of $S$ and $S'$ is $\{u_1,u_k\}$ as we argued before. 
Since $S$ is chordless, we know that no edge contains $\{u_{i-1},u_i,u_{i+1}\}$. Combining these facts, by \autoref{lemma:simple-cycle-implications} we obtain that $u_i$ is part of a chordless cycle. We can thus apply the induction hypothesis and are done.
\item
If there exists an edge between $v$ and some node on the path $f_{a+1},\ldots,f_{b-1}$, denote by $x$ and $y$ the smallest and largest indices between $1$ and $n$ such that $f_x$ and $f_y$ are neighbors of $v$ (it is possible that $x=y$). 
\begin{itemize}
    \item 
    If $v\in\free(Q_1')$:
        \begin{itemize}
        \item
        In case $v$ is not a neighbor of $f_1$ or not a neighbor of $f_n$, we can assume without loss of generality that $v$ is not a neighbor of $f_n$. Then the path $v,f_{y},\ldots,f_n$ is chordless. It consists of at least $3$ nodes because its end-points are not neighbors, and so it is a free-path containing $v$. Thus, we can apply the induction hypothesis and are done.
        \item
        If $v$ is a neighbor of both $f_1$ and $f_n$,
        consider the cycle $v,f_1,\ldots,f_n,v$.
        This is a simple cycle since $S$ and $S'$ intersect in exactly $\{u_1,u_k\}$ as argued before.
        Since the free-path is chordless, $f_1$ and $f_n$ are not neighbors, and so by \autoref{lemma:simple-cycle-implications}, $v$ is part of a chordless cycle (in which case we can apply the induction hypothesis and are done) or the center of a hand-fan, which in this case is a free-hand-fan. In the latter case, the Reduction Lemma applies by \autoref{lemma:free-hand-fan-reduction}.
        \end{itemize}
    \item
    If $v\not\in\free(Q_1')$:
        \begin{itemize}
        \item
        If $v$ is a neighbor of some node in $f_{b+1},\ldots,f_n$ or some node in $f_1,\ldots,f_{a-1}$,
        since we also know that $v$ has some neighbor in $f_{a+1},\ldots,f_{b-1}$, this means that $v$ has (at least) two non-adjacent neighbors on the free-path. Hence, $f_1,\ldots,f_{x},v,f_{y},f_n$ is a free-path, and we are done.
        \item
        If $v$ has no neighbors in $f_1,\ldots,f_{a-1}$ and $f_{b+1},\ldots,f_n$:
            \begin{itemize}
            \item
            If there are no edges from any node in $u_2,\ldots,u_{i-1}$ to any node in $f_{y},\ldots,f_n$ and symmetrically there are no edges from any node in $u_{i+1},\ldots,u_{k-1}$ to any node in $f_1,\ldots,f_{x}$,
            then consider the path $f_1,\ldots,f_a=u_1,\ldots,u_k=f_b,\ldots,f_n$.
            This path contains $u_i$, and it does not contain chords that cross between its two sides, that is, there are no chords between a node before $u_i$ to a node after $u_i$ on this path. Hence, its chordless shortening is a free-path that contains $u_i$.
            \item
            In case there is an edge from some node in $u_2,\ldots,u_{i-1}$ to some node in $f_{y},\ldots,f_n$ or an edge from some node in $u_{i+1},\ldots,u_{k-1}$ to some node in $f_1,\ldots,f_{x}$,
            we can assume without loss of generality that there is an edge $u_p-f_q$ from some node in $u_2,\ldots,u_{i-1}$ to some node in $f_{y},\ldots,f_n$.
            Apply \autoref{lemma:simple-cycle-implications} on the simple cycle $u_p,\ldots,u_k=f_b,\ldots,f_q,u_p$ and $v=u_i$.
            Since the cycle $S$ is chordless, there are no edges containing $\{u_{i-1},u_i,u_{i+1}\}$.
            Since $S$ is chordless and since $u_i$ has no neighbors in $f_{b+1},\ldots,f_n$, we get that $u_i$ has no neighbors in this cycle other than $u_{i-1}$ and $u_{i+1}$ (if $u_i$ has $f_b=u_k$ as a neighbor, we have that $i=k-1$), so $u_i$ is not the center of a hand-fan as described by \autoref{lemma:simple-cycle-implications}.
            We conclude that $v=u_i$ is part of a chordless cycle.
            \end{itemize}
        \end{itemize}
    \end{itemize}
\end{itemize}

In all cases, we either showed that $v$ is in a difficult structure in $Q_1'$, so we can use the induction hypothesis to show that the Reduction Lemma applies, or we directly showed the Reduction Lemma applies using \autoref{lemma:free-hand-fan-reduction}, \autoref{lemma:flower-reduction}, or \autoref{lemma:almost-tetra-reduction}.
This concludes the proof by induction and proves \autoref{lemma:general-arity-reduction-applies}.

\subsection{A VUTD-Based Dichotomy}\label{sec:dichotomy}

In this section, we show that if we assume the \UTD{} hypothesis, we can conclude the previously known hardness results without making additional assumptions.
\autoref{thm:two-hards-intractability} states that if a union of two difficult CQs does not admit a free-connex union extension, then it is not in $\DelayClin$ assuming the \hyperclique{} and \fourclique{} hypotheses.
To conclude with a dichotomy, it remains to replace the hypotheses used in \autoref{thm:two-hards-intractability} by the new \UTD{} hypothesis.
We show that assuming the \hyperclique{} hypothesis can always be replaced with assuming the \UTD{} hypothesis.

\begin{prop}\label{prop:hyperclique-UTD}
The \UTD{} hypothesis implies the \hyperclique{} hypothesis.
\end{prop}
\begin{proof}
If the \hyperclique{} hypothesis does not hold, there exists $k\geq 3$ such that it is possible to determine the existence of a $k$-hyperclique in a $(k-1)$-uniform hypergraph with $n$ vertices in time $O(n^{k-1})$.
Set $\alpha=\frac{1}{k-2}$.

Assume we are given a tripartite graph $G$ with vertex sets $V_1$,$V_2$,$V_3$ and edge sets $E_{1,2}$, $E_{2,3}$, $E_{1,3}$ where $|V_1|=|V_2|=\Theta(n^\alpha)$ and $|V_3|=n$.
We now construct a hyperclique instance~$G'$.
We encode the vertices of $V_3$ as $U_3\times\cdots\times U_k$ such that $|U_3|=\ldots=|U_{k}|=\Theta(n^\alpha)$.
We can do this by putting $\lceil\alpha\log{n}\rceil$ bits of the binary representation of an identifier of a vertex of $V_3$ in each of $U_3,\ldots U_{k}$.
For every edge $(v_1,v_3)\in E_{1,3}$, add an edge $\{v_1,u_3,\ldots,u_k\}$ where $u_3,\ldots,u_k$ is the representation of $v_3$.
For every edge $(v_2,v_3)\in E_{2,3}$, add an edge $\{v_2,u_3,\ldots,u_k\}$ where $u_3,\ldots,u_k$ is the representation of $v_3$.
For every edge $(v_1,v_2)\in E_{1,2}$, add an edge containing $v_1$, $v_2$ and every combination of $k-3$ vertices from distinct sets in $U_3,\ldots,U_k$.
This results in a $(k-1)$-uniform hypergraph $G'$ with $O(kn^\alpha)$ vertices, and
$k$-hypercliques in $G'$ are in one-to-one correspondence to triangles in the tripartite graph $G$.
Using the assumed algorithm, we can detect a $k$-hyperclique in $G'$ and thus a triangle in $G$ in time $O((kn^\alpha)^{k-1})=O(n^{1+\alpha})$, contradicting the \UTD{} hypothesis.
\end{proof}

The original statement of \autoref{thm:two-hards-intractability} also assumes the \BMM{} hypothesis~\cite{DBLP:conf/pods/CarmeliK19}, but this assumption is not required if the \hyperclique{} hypothesis is already assumed.

\begin{prop}\label{prop:BMM-UTD}
The \hyperclique{} hypothesis implies the \BMM{} hypothesis.
\end{prop}
\begin{proof}
Boolean matrix multiplication can be used to detect triangles in a tripartite graph:
Consider the multiplication of the adjacency matrix of $V_1$ and $V_2$ with the adjacency matrix of $V_2$ and $V_3$. Every result is a path of length two, and we can check in constant time whether its end-points are neighbors. Therefore, we can find all triangles in the same time it takes to multiply the matrices.
If \BMM{} does not hold, it is possible to multiply two Boolean $n\times n$ matrices in $O(n^2)$ time, and so it is possible to find triangles in a tripartite graph with $n$ vertices in time $O(n^2)$. This contradicts the \hyperclique{} hypothesis (for $k=3$).
\end{proof}

We do not know whether the \fourclique{} hypothesis is also implied by $\UTD{}$. Instead, we can show that the case in which Carmeli and Kröll~\cite{DBLP:conf/pods/CarmeliK19} use the \fourclique{} hypothesis can be resolved directly by our reduction in \autoref{lemma:reduction}. For this purpose, we need to recall the following structural properties.

\begin{defiC}[\cite{DBLP:conf/pods/CarmeliK19}]
Let $Q=Q_1\cup Q_2$ be a union of two body-isomorphic CQs, and let $h$ be a body-isomorphism from $Q_1$ to $Q_2$.
\begin{itemize}
\item $Q_1$ is said to be \emph{free-path guarded} if for every free-path $v_1,\ldots,v_k$ in $Q_1$, we have that $h(v_i)\in\free(Q_2)$ for all $1\le i\le k$.
\item $Q_1$ is said to be \emph{bypass guarded} if for every free-path $v_1,\ldots,v_k$ in $Q_1$ and variable $u$ of $Q_1$, if $Q_1$ has atoms containing $\{v_{i-1},v_i,u\}$ and $\{v_i,v_{i+1},u\}$ for some $1<i<k$, then $h(u)\in\free(Q_2)$.
\end{itemize}
\end{defiC}

We can now show that \autoref{thm:two-hards-intractability} holds independently of additional assumptions if we assume the \UTD{} hypothesis. \autoref{thm:intractable} also specifies the structural properties required for tractability.

\begin{thm}\label{thm:intractable}
Let $Q=Q_1\cup Q_2$ be a non-redundant union of two difficult CQs.
Assuming the \UTD{} hypothesis, the following are equivalent:
\begin{itemize}
    \item $Q\in\DelayClin$.
    \item $Q$ admits a free-connex union extension.
    \item $Q_1$ and $Q_2$ are body-isomorphic, acyclic, free-path guarded, and bypass guarded.
\end{itemize}
\end{thm}

\begin{proof}
If $Q_1$ and $Q_2$ are body-isomorphic, acyclic, free-path guarded, and bypass guarded, then $Q$ admits a free-connex union extension~\cite[Lemma 32]{DBLP:conf/pods/CarmeliK19}. According to \autoref{theorem:UCQtractability}, if $Q$ admits a free-connex union extension, then $Q\in\DelayClin$. It remains to prove that if $Q\in\DelayClin$, then $Q_1$ and $Q_2$ are body-isomorphic, acyclic, free-path guarded, and bypass guarded. We will show the contrapositive.
Assume the \UTD{} hypothesis. Then by Proposition \ref{prop:hyperclique-UTD} the \hyperclique{} hypothesis holds, which by \autoref{prop:BMM-UTD} also implies the \BMM{} hypothesis. Some of the results we cite next rely on these assumptions.

If $Q_1$ and $Q_2$ are not body-isomorphic and acyclic, then $Q\not\in\DelayClin$~\cite[Theorem 21]{DBLP:conf/pods/CarmeliK19}, so assume that they are body-isomorphic and acyclic.
If one of the CQs is not free-path guarded,
then $Q\not\in\DelayClin$~\cite[Lemma 27]{DBLP:conf/pods/CarmeliK19}.
It is left to handle the case that $Q_1$ and $Q_2$ are both free-path guarded and one of the CQs is not bypass guarded.
Assume without loss of generality that $Q_1$ is not bypass guarded.
We replace the hardness proof based on the \fourclique{} hypothesis~\cite[Lemma 28]{DBLP:conf/pods/CarmeliK19} with a hardness proof using our Reduction Lemma. Let $h$ be a body-isomorphism from $Q_1$ to $Q_2$.
In this case, we know there exist variables $z_0,z_1,z_2,u$ in $Q_1$ such that the following holds~\cite[in proof of Lemma 31]{DBLP:conf/pods/CarmeliK19}:
$z_0,z_2\in\free(Q_1)$,
$z_1\not\in\free(Q_1)$,
$h(u)\not\in\free(Q_2)$,
and the CQ $Q_1$ has two atoms containing $\{z_0,z_1,u\}$ and $\{z_1,z_2,u\}$
but no atom containing $\{z_0,z_2\}$. 
We can use the Reduction Lemma with $X_1=\{z_0\}$, $X_2=\{z_2\}$, $X_3=\{z_1\}$, and $X_4=\{u\}$.
Since there is no atom containing both $z_0$ and $z_2$, Condition~\ref{cond:notall} holds.
Condition~\ref{cond:connected} trivially holds since the sets $X_i$ are of size one. Condition~\ref{cond:connectors} holds due to the atoms containing $\{z_0,z_1,u\}$ and $\{z_1,z_2,u\}$, and since $z_0,z_2\in\free(Q_1)$, where the free variables form a valid connector since $z_1\not\in\free(Q_1)$. Condition~\ref{cond:otherquery} holds since $h(u)\not\in\free(Q_2)$.
Hence, $Q\not\in\DelayClin$.
\end{proof}

Combining \autoref{theorem:UCQtractability}, \autoref{thm:general-hardness} and \autoref{thm:intractable} allows us to base the entire dichotomy on one hypothesis.

\begin{cor}\label{cor:dichotomy}
Let $Q$ be a non-redundant union of two self-join-free CQs.
Assuming the \UTD{} hypothesis, the following are equivalent:
\begin{itemize}
    \item $Q\in\DelayClin$.
    \item $Q$ admits a free-connex union extension.
    \item At least one of the following holds:\begin{itemize}
        \item Both CQs are free-connex.
        \item The CQs are body-isomorphic, acyclic, free-path guarded, and bypass guarded.
        \item One CQ is free-connex, and the resolution process of Definition~\ref{def:resolution} turns the second CQ free-connex. 
    \end{itemize}
\end{itemize}
\end{cor}

\begin{proof}
In case both CQs are free-connex, then $Q$ is a free-connex union extension of itself, and $Q\in\DelayClin$  according to \autoref{theorem:UCQtractability}. Assume next that at least one of the CQs is not free-connex.
In case both CQs are difficult, by \autoref{thm:intractable}, $Q$ is tractable and admits a free-connex union extension if and only if the CQs are body-isomorphic, acyclic, free-path guarded, and bypass guarded.
The remaining case is that one CQ is free-connex and the other is difficult. In this case, \autoref{thm:general-hardness} shows that $Q$ is tractable and admits a free-connex union extension if and only if the resolved query is free-connex.
\end{proof}

\section{Discussions}

In the previous section, we proved that union extensions capture all unions of two self-join-free CQs in $\DelayClin$, assuming the \UTD{} hypothesis. We next discuss the possibility of showing variations of this result. In \autoref{sec:super-const} we discuss the relaxation of our tractability notion by allowing logarithmic factors, and in \autoref{appendix:otherUTDhypothesis} we explain why we cannot use a different hypothesis reminiscent of \UTD{}.

\subsection{Super-Constant Delay}\label{sec:super-const}

The class $\DelayClin$ is quite restrictive in that the preprocessing time must be linear and the delay must be constant. A natural relaxation of this class allows near-linear preprocessing time and polylogarithmic delay. As it turns out, our results also apply to this relaxed class, if we replace our \UTD{} hypothesis by the following.

\begin{description}
 \item[Stronger \UTD{} (sVUTD) Hypothesis] For any constant $\alpha\in(0,1]$ there exists $\varepsilon > 0$ such that determining whether there exists a triangle in a tripartite graph with $|V_3|=n$ and $|V_1|=|V_2|=\Theta(n^\alpha)$
cannot be done in time $O(n^{1+\alpha+\varepsilon})$.
\end{description}

For $\alpha = 1$ this hypothesis postulates that the exponent of matrix multiplication is $\omega > 2$.\footnote{Note that $\omega > 2$ implies the \BMM{} hypothesis but is strictly stronger, as e.g.\ an $O(n^2 \log n)$-time algorithm for BMM would show $\omega = 2$ but does not falsify the \BMM{} hypothesis.}
The case $\alpha < 1$ is an unbalanced analog of this. Hence, by essentially the same discussion as in \autoref{sec:utd-definition}, sVUTD formalizes a computational barrier.

Our proof can be modified to use sVUTD instead of the \UTD{} hypothesis.
With the exception of \autoref{sec:dichotomy} (that cites hardness results proved previously), our hardness results use the Reduction Lemma, so it is enough to modify that lemma. The proof of the Reduction Lemma shows how to build a database instance of size $O(n^{1+\alpha})$ that allows detecting triangles in an $\alpha$-unbalanced graph by enumerating at most $O(n^{1+\alpha})$ query answers. Thus, a query answering algorithm that runs with $O(|I|^{1+\varepsilon})$ preprocessing time and $O(|I|^\varepsilon)$ delay can detect triangles in $O(n^{(1+\alpha)(1+\varepsilon)})\leq O(n^{1+\alpha+2\varepsilon})$ time, which contradicts the sVUTD hypothesis for sufficiently small $\varepsilon$.
Thus, whenever the Reduction Lemma can be applied, we obtain that the UCQ cannot be solved with near-linear preprocessing time and polylogarithmic delay assuming the sVUTD hypothesis.
Similar variations can be applied to the results of \autoref{sec:dichotomy} by showing that sVUTD implies stronger versions of the \hyperclique{} and \BMM{} hypotheses, postulating that $k$-hypercliques cannot be detected in $O(n^{k-1+\varepsilon})$ time and that Boolean matrices cannot be multiplied in $O(n^{2+\varepsilon})$ time.

With these modifications, we obtain the following:
\emph{For every UCQ $Q$ for which we have shown $Q \not\in \DelayClin$ assuming the \UTD{} hypothesis, there exists a constant $\varepsilon > 0$ such that $Q$ cannot be answered with $O(|I|^{1+\varepsilon})$ preprocessing time and $O(|I|^\varepsilon)$ delay on input database $I$,} assuming the sVUTD hypothesis.
Thus, if we assume sVUTD, free-connex union extensions also capture all unions of two self-join-free CQs that can be answered with near-linear preprocessing time and polylogarithmic delay.

\subsection{Difference From Another UTD Hypothesis}
\label{appendix:otherUTDhypothesis}
A hypothesis called Unbalanced Triangle Detection was recently formulated by Kopelowitz and Vassilevska Williams~\cite{KopelowitzW20}.
In order to differentiate, we will refer to their hypothesis as \emph{Edge-Unbalanced Triangle Detection} (EUTD). Their hypothesis states the following:

\begin{description}
\item[(Edge-)Unbalanced Triangle Detection Hypothesis (EUTD)~\cite{KopelowitzW20}]
Given any constants $0 < \alpha \le \beta \le 1$ and $\varepsilon>0$, determining whether there exists a triangle in an $m$-edge tripartite graph with $O(m^\alpha)$ edges between $V_1$ and $V_2$ and $O(m^\beta)$ edges between $V_2$ and $V_3$
 has no algorithm running in time $O(m^{2/3+(\alpha+\beta)/3 - \varepsilon})$.
\end{description}

Note that the EUTD unbalancedness property restricts the number of edges between $V_1$ and $V_2$ and between $V_2$ and $V_3$, while our \UTD{} hypothesis restricts the number of vertices $|V_1|, |V_2|$. In the following, we discuss the reasons why we cannot use EUTD instead of \UTD{} in our work.

First, there is an algorithm for edge-unbalanced triangle detection that runs in time $O(m^{2/3 + (\alpha+\beta)/3} + m)$ if the exponent of matrix multiplication is $\omega = 2$~\cite[Theorem~5]{KopelowitzW20}. This matches the EUTD hypothesis (note that the additive term $O(m)$ is necessary to read the input). Intuitively, this means that EUTD is not strong enough to imply any lower bound on matrix multiplication, since assuming both the EUTD hypothesis and $\omega=2$ does not lead to a contradiction (at least not immediately). However, our proof requires hardness of BMM (see \autoref{sec:dichotomy}), so EUTD is not sufficiently strong for our purposes. 

Second, in one case of our proof, we argue that no algorithm can list all pairs of vertices in $V_1 \times V_3$ that are connected by a 2-path in linear time in terms of the input plus output size (cf.~the case of a free-path where an end variable is not provided in the proof of \autoref{lemma:original-apply}). This is implied by our \UTD{} hypothesis since there are $O(n^{1+\alpha})$ pairs of vertices in $V_1 \times V_3$, so the input plus output size is $O(n^{1+\alpha})$, and thus any such algorithm would contradict the \UTD{} hypothesis.
However, in the setting of EUTD, the number of pairs in $V_1 \times V_3$ that are connected by a 2-path can be up to $\Omega(m^{1+\alpha})$ (for $\beta=1$). Since this is much larger than the running time lower bound postulated by EUTD, the EUTD hypothesis does not say anything about the problem of listing pairs in $V_1 \times V_3$ connected by a 2-path.

\section{Conclusion}\label{sec:conclusion}

In this paper, as a first step, we proved new conditional lower bounds for UCQ answering based on the \threesum{} conjecture, via \UTL{}. Next, we defined the \UTD{} hypothesis and showed that it is closely connected to UCQ answering: On the one hand, we presented examples of UCQs that are currently not known to be in $\DelayClin$ but would be in $\DelayClin$ if the \UTD{} hypothesis is false. On the other hand, we used this hypothesis to establish a dichotomy for a class of UCQs: if the \UTD{} hypothesis is true, a union of two self-join-free CQs is in $\DelayClin$ if and only if it has a free-connex union extension.
Moreover, we identified the structural properties of the UCQs that fall within this tractable class.
Overall, we reduced a question about many UCQs to an arguably simpler question about unbalanced triangle detection: in order to reason about whether there exist UCQs in $\DelayClin$ that do not have a free-connex union extension, we should inspect the \UTD{} hypothesis.
If we assume the \UTD{} hypothesis, then the answer is `no' when considering unions of two self-join-free CQs.
If, on the other hand, we find a linear time algorithm for \UTD{}, then the answer is `yes', and we obtain a linear preprocessing and constant delay algorithm for additional UCQs.
This paper resolved all example UCQs that were stated as open in previous work except for one example of a union of more than two CQs (\cite[Example 36]{DBLP:conf/pods/CarmeliK19}).
Natural next steps are to try and prove a dichotomy for unions of more than two CQs (which seems very challenging, as we do not currently know how to resolve the aforementioned open example), and to study the unbalanced triangle detection problem further.

\section*{Acknowledgment}
  \noindent  \emph{Karl Bringmann:} This work is part of the project TIPEA that has received funding from the European Research Council (ERC) under the European Unions Horizon 2020 research and innovation programme (grant agreement No.~850979).
 \emph{Nofar Carmeli:} This work was supported by the Google PhD Fellowship. It was also funded by the French government under management of Agence Nationale de la Recherche as part of the ``Investissements d’avenir'' program, reference ANR-19-P3IA-0001 (PRAIRIE 3IA Institute).

\bibliographystyle{alphaurl}
\bibliography{references}

\begin{thebibliography}{WXXZ24}

\bibitem[AFLG15]{AmbainisFG15}
Andris Ambainis, Yuval Filmus, and Fran{\c{c}}ois Le~Gall.
\newblock Fast matrix multiplication: Limitations of the
  {C}oppersmith-{W}inograd method.
\newblock In {\em {STOC}}, pages 585--593. {ACM}, 2015.

\bibitem[ASU13]{AlonSU13}
Noga Alon, Amir Shpilka, and Christopher Umans.
\newblock On sunflowers and matrix multiplication.
\newblock {\em Comput. Complex.}, 22(2):219--243, 2013.

\bibitem[AVW18]{AlmanVW18}
Josh Alman and Virginia Vassilevska~Williams.
\newblock Further limitations of the known approaches for matrix
  multiplication.
\newblock In {\em {ITCS}}, volume~94 of {\em LIPIcs}, pages 25:1--25:15, 2018.

\bibitem[AW23]{VWilliams19}
Josh Alman and Virginia~Vassilevska Williams.
\newblock Limits on all known (and some unknown) approaches to matrix
  multiplication.
\newblock {\em {SIAM} J. Comput.}, 52(6):S18--285, 2023.

\bibitem[BB13]{bb:thesis}
Johann Brault-Baron.
\newblock {\em De la pertinence de l’{\'e}num{\'e}ration: complexit{\'e} en
  logiques propositionnelle et du premier ordre}.
\newblock PhD thesis, Universit{\'e} de Caen, 2013.

\bibitem[BB16]{brault2016hypergraph}
Johann Brault-Baron.
\newblock Hypergraph acyclicity revisited.
\newblock {\em ACM Computing Surveys (CSUR)}, 49(3):1--26, 2016.

\bibitem[BDG07]{bdg:dichotomy}
Guillaume Bagan, Arnaud Durand, and Etienne Grandjean.
\newblock On acyclic conjunctive queries and constant delay enumeration.
\newblock In {\em International Workshop on Computer Science Logic}, pages
  208--222. Springer, 2007.

\bibitem[BGS20]{berkholz2020tutorial}
Christoph Berkholz, Fabian Gerhardt, and Nicole Schweikardt.
\newblock Constant delay enumeration for conjunctive queries: a tutorial.
\newblock {\em ACM SIGLOG News}, 7(1):4--33, 2020.

\bibitem[Bri19]{Bringmann19}
Karl Bringmann.
\newblock Fine-grained complexity theory (tutorial).
\newblock In {\em 36th International Symposium on Theoretical Aspects of
  Computer Science (STACS 2019)}. Schloss Dagstuhl-Leibniz-Zentrum fuer
  Informatik, 2019.

\bibitem[CH20]{ChanH20}
Timothy~M. Chan and Qizheng He.
\newblock Reducing {3SUM} to {Convolution-3SUM}.
\newblock In {\em {SOSA}}, pages 1--7. {SIAM}, 2020.

\bibitem[CJN18]{icdt/BerkholzKS18}
Berkholz Christoph, Keppeler Jens, and Schweikardt Nicole.
\newblock Answering ucqs under updates and in the presence of integrity
  constraints.
\newblock In {\em Proceedings of the 21st International Conference on Database
  Theory (ICDT’18)}, volume~98, pages 1--8, 2018.

\bibitem[CK21]{DBLP:conf/pods/CarmeliK19}
Nofar Carmeli and Markus Kr{\"o}ll.
\newblock On the enumeration complexity of unions of conjunctive queries.
\newblock {\em ACM Transactions on Database Systems (TODS)}, 46(2):1--41, 2021.

\bibitem[CM77]{homomorphism/redundancy}
Ashok~K Chandra and Philip~M Merlin.
\newblock Optimal implementation of conjunctive queries in relational data
  bases.
\newblock In {\em Proceedings of the ninth annual ACM symposium on Theory of
  computing}, pages 77--90, 1977.

\bibitem[CS23]{carmeli2023selfjoins}
Nofar Carmeli and Luc Segoufin.
\newblock Conjunctive queries with self-joins, towards a fine-grained
  enumeration complexity analysis.
\newblock In {\em Proceedings of the 42nd ACM SIGMOD-SIGACT-SIGAI Symposium on
  Principles of Database Systems}, pages 277--289, 2023.

\bibitem[CW90]{CoppersmithW90}
Don Coppersmith and Shmuel Winograd.
\newblock Matrix multiplication via arithmetic progressions.
\newblock {\em J. Symb. Comput.}, 9(3):251--280, 1990.
\newblock \href {https://doi.org/10.1016/S0747-7171(08)80013-2}
  {\path{doi:10.1016/S0747-7171(08)80013-2}}.

\bibitem[DS11]{csl/DurandS11}
Arnaud Durand and Yann Strozecki.
\newblock Enumeration complexity of logical query problems with second-order
  variables.
\newblock In {\em Computer Science Logic (CSL'11)-25th International
  Workshop/20th Annual Conference of the EACSL (2011)}.
  Schloss-Dagstuhl-Leibniz Zentrum f{\"u}r Informatik, 2011.

\bibitem[Dur20]{DurandTutorial20}
Arnaud Durand.
\newblock Fine-grained complexity analysis of queries: From decision to
  counting and enumeration.
\newblock In {\em Proceedings of the 39th ACM SIGMOD-SIGACT-SIGAI Symposium on
  Principles of Database Systems}, pages 331--346, 2020.

\bibitem[FKP24]{FischerK024}
Nick Fischer, Piotr Kaliciak, and Adam Polak.
\newblock Deterministic {3SUM}-hardness.
\newblock In {\em {ITCS}}, volume 287 of {\em LIPIcs}, pages 49:1--49:24.
  Schloss Dagstuhl - Leibniz-Zentrum f{\"{u}}r Informatik, 2024.

\bibitem[GO95]{gajentaan1995class}
Anka Gajentaan and Mark~H Overmars.
\newblock On a class of {$O(n^2)$} problems in computational geometry.
\newblock {\em Computational geometry}, 5(3):165--185, 1995.

\bibitem[GU18]{GallU18}
Francois~Le Gall and Florent Urrutia.
\newblock Improved rectangular matrix multiplication using powers of the
  {Coppersmith-Winograd} tensor.
\newblock In {\em Proceedings of the Twenty-Ninth Annual ACM-SIAM Symposium on
  Discrete Algorithms}, pages 1029--1046. SIAM, 2018.

\bibitem[IR77]{ItaiR78}
Alon Itai and Michael Rodeh.
\newblock Finding a minimum circuit in a graph.
\newblock In {\em Proceedings of the ninth annual ACM symposium on Theory of
  computing}, pages 1--10, 1977.

\bibitem[KPP16]{DBLP:conf/soda/KopelowitzPP16}
Tsvi Kopelowitz, Seth Pettie, and Ely Porat.
\newblock Higher lower bounds from the {3SUM} conjecture.
\newblock In {\em Proceedings of the twenty-seventh annual ACM-SIAM symposium
  on Discrete algorithms}, pages 1272--1287. SIAM, 2016.

\bibitem[KVW20]{KopelowitzW20}
Tsvi Kopelowitz and Virginia Vassilevska~Williams.
\newblock Towards optimal set-disjointness and set-intersection data
  structures.
\newblock In {\em 47th International Colloquium on Automata, Languages, and
  Programming (ICALP 2020)}. Schloss Dagstuhl-Leibniz-Zentrum f{\"u}r
  Informatik, 2020.

\bibitem[Pat10]{Patrascu10}
Mihai Patrascu.
\newblock Towards polynomial lower bounds for dynamic problems.
\newblock In {\em Proceedings of the forty-second ACM symposium on Theory of
  computing}, pages 603--610, 2010.

\bibitem[Var82]{vardi1982complexity}
Moshe~Y Vardi.
\newblock The complexity of relational query languages.
\newblock In {\em Proceedings of the 14th Annual ACM Symposium on Theory of
  Computing}, pages 137--146, 1982.

\bibitem[VWX20]{WilliamsX20}
Virginia Vassilevska~Williams and Yinzhan Xu.
\newblock Monochromatic triangles, triangle listing and {APSP}.
\newblock In {\em {FOCS}}, pages 786--797. {IEEE}, 2020.

\bibitem[Wil18]{williams2018some}
Virginia~Vassilevska Williams.
\newblock On some fine-grained questions in algorithms and complexity.
\newblock In {\em Proceedings of the international congress of mathematicians:
  Rio de janeiro 2018}, pages 3447--3487. World Scientific, 2018.

\bibitem[WXXZ24]{WilliamsXXZ24}
Virginia~Vassilevska Williams, Yinzhan Xu, Zixuan Xu, and Renfei Zhou.
\newblock New bounds for matrix multiplication: from alpha to omega.
\newblock In {\em {SODA}}, pages 3792--3835. {SIAM}, 2024.

\end{thebibliography}

\appendix
\section{Proofs for \autoref{sec:completeness}}\label{sec:binary-proof}

We first prove some lemmas needed for the proof of \autoref{lemma:binary-free-connex}.

\begin{lem}\label{lemma:orig-connected}
Let $Q=Q_1\cup Q_2$ where $Q_1$ is self-join-free, and let $Q_1^+$ be the resolved $Q_1$.
If there is a path $P^+$ between $u$ and $v$ in $Q_1^+$, then there is a chordless path between $u$ and $v$ in $Q_1$ that goes only through variables of $\var(P^+)\cup h(\free(Q_2))$, where $h$ is the unique body-homomorphism from $Q_2$ to $Q_1$.
\end{lem}

\begin{proof}
Every edge in $Q_1^+$ either: (1) is an edge of $Q_1$; or (2) contains the variables of a difficult structure with variables contained in $h(\free(Q_2))$.
Note that every difficult structure is connected.
First, we obtain a path in $Q_1$ that starts and ends in the same variables as $P^+$, by replacing every new edge of $Q_1^+$ in $P^+$ with a corresponding path through the difficult structure that it covers. Then we take a chordless path contained in this path.
\end{proof}

\begin{lem}\label{lemma:orig-free-end}
Let $Q=Q_1\cup Q_2$ where $Q_1$ is self-join-free, and let $Q_1^+$ be the resolved $Q_1$.
If there is a path $P^+$ in $Q_1^+$ from a variable $v$ to a variable in $\free(Q_1)$, then there is a chordless path $P$ in $Q_1$ from $v$ to some $u\in\free(Q_1)$ such that $\var(P)\cap\free(Q_1)=\{u\}$, and $\var(P)\subseteq\var(P^+)\cup h(\free(Q_2))$, where $h$ is the unique body-homomorphism from $Q_2$ to $Q_1$.
\end{lem}

\begin{proof}
First, take the chordless path $P'$ in $Q_1$ that is obtained from $P^+$ using \autoref{lemma:orig-connected}.
Then, take the subpath of $P'$ between $v$ and the first variable in $\free(Q_1)$. Such a variable exists because $P'$ ends in a free variable.
\end{proof}

\begin{lem}\label{lemma:binary-cycle-to-chordless}
If a vertex $v$ appears in a simple cycle in a graph, then $v$ also appears in a chordless cycle.
\end{lem}

\begin{proof}
Denote the cycle by $v,v_2,\ldots, v_m,v$.
Take a chordless path contained in $v_2,\ldots, v_m$, denote it $v_2=u_1,\ldots,u_k=v_m$.
Let $u_t$ be the first vertex after $u_1$ which is a neighbor of $v$. Such a vertex exists because $u_m$ is a neighbor.
Then, the cycle $v,u_1,\ldots,u_t,v$ is chordless.
\end{proof}

\autoref{lemma:binary-cycle-to-chordless} may seem trivial for graphs, but a similar statement does not hold for hypergraphs (\autoref{lemma:simple-cycle-implications} is the equivalent, more complicated, lemma for hypergraphs). In fact, this difference between graphs and hypergraphs is the main reason why we cannot show \autoref{lemma:binary-free-connex} for UCQs with general relations (of arity larger than $2$).
We can now prove \autoref{lemma:binary-free-connex}.

\begin{proof}[Proof of \autoref{lemma:binary-free-connex}]
Let $Q_1^+$ be the resolved $Q_1$, and assume for the sake of contradiction that $Q_1^+$ is not free-connex. Thus, it contains a difficult structure $S$.
Since $Q_2$ provides all difficult structures of $Q_1$, by construction, $Q_1^+$ has no difficult structures that also appear in $Q_1$. Thus, $S$ is a new difficult structure (that does not appear in $Q_1$).
By \autoref{prop:one-homo}, there is a single body-homomorphism $h$ from $Q_2$ to $Q_1$, and so the variables that $Q_2$ can provide to $Q_1$ are $h(\free(Q_2))$. Since $Q_1^+$ is resolved, the variables of $S$ are not contained in $h(\free(Q_2))$. We distinguish three cases according to the type of difficult structure $S$ is: a tetra of size $k>3$, a chordless cycle, or a free-path.

The first case is that $S$ is a tetra of size $k>3$.
Since some variable of the tetra is not in $h(\free(Q_2))$, all atoms of $Q_1^+$ that contain this variable appear in $Q_1$.
These atoms are therefore binary, which implies $k=3$, a contradiction to the assumption $k > 3$.

We now treat the case that $S$ is a chordless cycle. Denote the cycle by $x_1,\ldots,x_k$ such that $x_k\not\in h(\free(Q_2))$. Note that $x_{k-1},x_k,x_1$ are distinct variables.
Since $x_k$ is not provided, we know that the edges containing $\{x_{k-1}, x_k\}$ and $\{x_k, x_1\}$ appear in $Q_1$.
Due to the path $x_1,\ldots,x_{k-1}$ and since $x_k\not\in h(\free(Q_2))$, it follows from \autoref{lemma:orig-connected} that there is a simple path between $x_1$ and $x_{k-1}$ in $Q_1$ that does not go through $x_k$. This, together with the two edges $\{x_{k-1}, x_k\}$ and $\{x_k, x_1\}$, results in a simple cycle $x_1,\ldots,x_{k-1},x_k$ in $Q_1$.
By \autoref{lemma:binary-cycle-to-chordless}, since $x_k$ appears in a simple cycle in $Q_1$, it also appears in a chordless cycle in $Q_1$.
Since $x_k\not\in h(\free(Q_2))$, this contradicts our assumption that all difficult structures are provided.

Finally, we consider the case that $S$ is a free-path and denote $S=x_1,\ldots,x_k$. We have that $x_j\not\in h(\free(Q_2))$ for some $1\leq j \leq k$.
We now prove that $x_j$ appears in a difficult structure in $Q_1$.
This would mean that $x_j\in h(\free(Q_2))$, which is a contradiction.
First, assume that $x_j$ is at an end of the path; without loss of generality, $j=1$.
Since $x_1\not\in h(\free(Q_2))$, every edge containing $x_1$ in $Q_1^+$ also appears in $Q_1$, and so there is an edge $\set{x_1,x_2}$ in $Q_1$.
Since $x_2,\ldots,x_k$ is a path in $Q_1^+$ and $x_k\in\free(Q_1)$, by \autoref{lemma:orig-free-end}, there is a chordless path $x_2=t_1,\ldots,t_m$ in $Q_1$ such that $t_m$ is the only variable in $\{t_1,\ldots,t_m\}\cap\free(Q_1)$ and $\{t_1,\ldots,t_m\}\subseteq\{x_2,\ldots,x_k\}\cup h(\free(Q_2))$.
Note that this path does not contain $x_1$, and
it is a chordless path of length $1$ or more that ends with a free variable, and all other variables are not free.
If there is a neighbor $t_i$ of $x_1$ with $i>1$, take $i$ to be the minimal such index, and $x_1,t_1,\ldots,t_i,x_1$ is a chordless cycle.
Otherwise, $x_1,t_1,\ldots,t_m$ is a chordless path, and it is a free-path.

We now address the case that $1<j<k$.
Apply the same process as before (\autoref{lemma:orig-free-end}) on both sides of $S$ to obtain chordless simple paths $x_{j+1}=t_1,\ldots,t_m$ and $x_{j-1}=v_1,\ldots,v_n$ that do not contain $x_j$, where $v_n$ and $t_m$ are free, and the other variables are existential.
Note that since $x_{j-1},x_j,x_{j+1}$ is part of a chordless path in $Q_1$, these three variables are distinct.
If the two paths share a variable or neighbors, $x_j$ is part of a chordless cycle. Otherwise, $v_n,\ldots,x_{j-1},x_j,x_{j+1},\ldots,t_m$ is a free-path.
\end{proof}

\autoref{thm:binary-hardness} now easily follows.

\begin{proof}[Proof of \autoref{thm:binary-hardness}]
If the resolution process of Definition~\ref{def:resolution} turns $Q_1$ free-connex, the resolved UCQ is a free-connex union extension of $Q$, and according to \autoref{theorem:UCQtractability}, $Q\in\DelayClin$.
If the resolved $Q_1^+$ is not free-connex, by \autoref{lemma:binary-free-connex}, $Q_2$ does not provide all difficult structures in $Q_1$.
According to \autoref{lemma:original-apply}, the Reduction Lemma can be applied, and $Q\not\in\DelayClin$, assuming the \UTD{} hypothesis.
\end{proof}

\end{document}